\documentclass[12pt,a4paper,reqno]{amsart}
\usepackage[hidelinks]{hyperref}
\usepackage{amsmath,amsthm,amssymb}
\usepackage{amscd}
\usepackage{enumerate}
\usepackage{siunitx}
\usepackage{tikz-cd}
\usepackage{bm}
\usepackage{dsfont}

\usepackage[utf8]{inputenc}
\usepackage[T1]{fontenc}
\usepackage{lmodern}
\usepackage[english]{babel}
\usepackage{microtype}              
\usepackage{tikz}

\numberwithin{equation}{section}

\usepackage{mathtools}
\usepackage[tableposition=top]{caption}
\usepackage{booktabs,dcolumn}

\theoremstyle{plain}

\newtheorem{theorem}{Theorem}[section]
\newtheorem{proposition}[theorem]{Proposition}
\newtheorem{lemma}[theorem]{Lemma}
\newtheorem{corollary}[theorem]{Corollary}

\newtheorem{main_theorem}{Theorem}

\theoremstyle{definition}

\newtheorem{definition}[theorem]{Definition}
\newtheorem{remark}[theorem]{Remark}

\newtheorem{main_remark}{Remark}

\makeatletter
\def\@tocline#1#2#3#4#5#6#7{\relax
  \ifnum #1>\c@tocdepth 
  \else
    \par \addpenalty\@secpenalty\addvspace{#2}%
    \begingroup \hyphenpenalty\@M
    \@ifempty{#4}{%
      \@tempdima\csname r@tocindent\number#1\endcsname\relax
    }{%
      \@tempdima#4\relax
    }%
    \parindent\z@ \leftskip#3\relax \advance\leftskip\@tempdima\relax
    \rightskip\@pnumwidth plus4em \parfillskip-\@pnumwidth
    #5\leavevmode\hskip-\@tempdima
      \ifcase #1
       \or\or \hskip 1em \or \hskip 2em \else \hskip 3em \fi%
      #6\nobreak\relax
      \dotfill
      \hbox to\@pnumwidth{\@tocpagenum{#7}}
    \par
    \nobreak
    \endgroup
  \fi}
\makeatother

\newcommand\dist{{\operatorname{dist}}}

\newcommand\eps{\varepsilon}

\mathtoolsset{centercolon}

\makeatletter
\providecommand*{\dif}%
{\@ifnextchar^{\DIfF}{\DIfF^{}}}
\def\DIfF^#1{%
    \mathop{\mathrm{\mathstrut d}}%
    \nolimits^{#1}\gobblespace
}
\def\gobblespace{%
    \futurelet\diffarg\opspace}
\def\opspace{%
    \let\DiffSpace\!%
    \ifx\diffarg(%
    \let\DiffSpace\relax
    \else
    \ifx\diffarg[%
    \let\DiffSpace\relax
    \else
    \ifx\diffarg\{%
    \let\DiffSpace\relax
    \fi\fi\fi\DiffSpace}
\makeatother

\newcommand{\Ccal}{\mathcal{C}}

\newcommand{\Tcal}{\mathcal{T}}

\newcommand{\Cbb}{\mathbb{C}}

\newcommand{\Rbb}{\mathbb{R}}
\newcommand{\Sbb}{\mathbb{S}}
\newcommand{\Zbb}{\mathbb{Z}}

\DeclareMathOperator{\Tr}{Tr}

\let\eps\varepsilon
\let\phi\varphi

\newcommand{\intR}{\int_{\mathbb{R}^3}}

\newcommand{\rhob}{\underline{\rho}}

\newcommand{\brho}{\bar{\rho}}

\newcommand{\indvar}[2]{\mathds{1}_{#1} \ast \eta_{#2}}
\newcommand{\indld}{\mathds{1}_{\ell\Delta} \ast \eta_\delta}
\newcommand{\indldt}{\mathds{1}_{t\ell\Delta} \ast \eta_{t\delta}}

\newcommand{\elda}{e_{\mathrm{LDA}}}

\newcommand{\clt}{c_{\mathrm{LT}}}

\DeclareMathOperator{\supp}{supp}
\DeclareMathOperator{\spec}{spec}

\newcommand{\hti}{\widetilde{h}}

\newcommand{\type}{paper}

\begin{document}

\title[The Validity of the LDA for Smooth Short Range Interaction Potentials]{The Validity of the Local Density Approximation for Smooth Short Range Interaction Potentials}

\author{Nicco Mietzsch}
\address{Department of Mathematics, LMU Munich, Theresienstrasse 39, 80333 Munich, Germany}
\email{nicco.mietzsch@campus.lmu.de}

\date{\today}

\begin{abstract}
{In the full quantum theory, the energy of a many-body quantum system with a given one-body density is described by the Levy-Lieb functional. It is exact, but very complicated to compute. For practical computations, it is useful to introduce the Local Density Approximation which is based on the local energy of constant densities. The aim of this \type{} is to make a rigorous connection between the Levy-Lieb functional theory and the Local Density Approximation. Our justification is valid for fermionic systems with a general class of smooth short range interaction potentials, in the regime of slowly varying densities. We follow a general approach developed by Lewin, Lieb and Seiringer for Coulomb potential \cite{Lewin_2020}, but avoid using any special properties of the potential including the scaling property and screening effects for the localization of the energy. }
\end{abstract}

\maketitle

\setcounter{tocdepth}{2}
\tableofcontents

\parindent 0mm
\parskip   5mm 

\section{Introduction} \label{chapter:intro}
{Even basic properties of many-body quantum systems are very difficult to compute for high numbers of particles. This is due to the high dimensionality of the respective equations. Therefore, there have been many attempts at reducing the computational cost by approximating the equations of the full system with simpler equations in lower dimensions which can be solved numerically or even analytically.

In the following, we will be interested in approximating the energy of a full many-body system by an explicit functional. We will start with describing our setting. In suitable units, a system of $N$ fermionic particles in $d$ dimensions without spin can be described by the Hamilton operator
\begin{equation}
	H_N^{V,\omega} = \sum_{i=1}^N \Delta_{x_i} + \sum_{i=1}^N V(x_i) + \sum_{1\leq j < k \leq N} \omega(x_j - x_k).
\end{equation}
Here, the functions $V$, $\omega$ describe the external potential and the interaction potential between two particles. Because we want to describe fermionic particles, this operator acts in the antisymmetric space $\bigwedge^N L^2(\Rbb^d)$. More precisely, under suitable assumptions on $\upsilon$ and $\omega$, the operator $H_N^{V,\omega}$ can be realized as a self-adjoint operator by Friedrichs' method. In particular, it is densely defined with domain $D(H_N^{V,\omega})$ and a quadratic form domain $Q(H_N^{V,\omega})$, which is the natural space of wavefunctions $\psi$ for which the energy $\left\langle \psi, H_N^{V,\omega} \psi \right\rangle$ is defined.

This operator encodes all relevant physical properties of the system in question. For example the \textit{ground state energy} is given by the variational principle
\begin{equation}
\label{eq:ground_state}
	E_N^{V,\omega} = \inf \spec (H_N^{V,\omega}) = \inf_{\substack{\psi \in Q(H_N^{V,\omega}) \\ \lVert \psi \rVert = 1}} \left\langle \psi, H_N^{V,\omega} \psi \right\rangle
\end{equation}
the infimum of the spectrum of $H_N^{V,\omega}$.

While these equations are easy to state from a mathematical point of view (after all, the objects involved are linear operators), in applications it is difficult to compute their solutions. The problem is that typically, one is interested in large systems, meaning one has to calculate solutions in a very high-dimensional space. This is difficult even numerically. Therefore, since the invention of quantum mechanics, there have been many attempts to reduce the computational cost by replacing these high-dimensional \textit{linear} equations in $N$-body space by low-dimensional \textit{non-linear} equations which have very few variables. However, while the resulting formulae have great success in applications, most of them have very unsatisfactory mathematical justifications. An important example is the Local Density Approximation in Density Functional Theory, which will be the topic of this \type{}.}

\subsection{Density Functional Theory}\label{sec:dft}
{To discuss the formulation of Density Functional Theory (DFT), we first need to define the density. While this will be done in greater generality in Section \ref{chapter:main_res}, here we restrict ourself to densities of wavefunctions $\psi \in \bigwedge_1^N L^2(\Rbb^d)$ with $\lVert \psi \lVert = 1$. In this case the one-body density $\rho_\psi$ is given by
\begin{equation*}
	\rho_\psi(x) = N \int_{\Rbb^{d(N-1)}} |\psi(x,x_2,\dots,x_n)|^2 \;dx_2 \cdots dx_n.
\end{equation*}
The interpretation of the density is that for any set $\Omega$, $\int_{\Omega} \rho_\psi$ gives the average number of particles in $\Omega$. In particular, we have $\int_{\Rbb^d} \rho_\psi = N$.

The idea of DFT is to split the minimization problem (\ref{eq:ground_state}) into two minimizations, first over all wavefunctions having a particular density and then over all possible densities
\begin{equation}
\label{eq:dft}
	E_N^{V,\omega} = \inf_\rho \left\lbrace \inf_{\rho_\psi = \rho} \left\langle \psi, H_N^{V,\omega} \psi \right\rangle \right\rbrace
\end{equation}
where the inner product is understood in the sense of quadratic forms. By the (anti-) symmetry of the wavefunction, it is easy to see that
\begin{equation*}
	\left\langle \psi, \sum_{i=1}^N V(x_i) \psi \right\rangle = \int_{\Rbb^d} V(x) \rho_\psi(x) \; dx.
\end{equation*}
Therefore, it suffices to minimize the kinetic and interaction energy, i.e. to define the Levy-Lieb functional \cite{Levy_1979, Lieb_1983}
\begin{equation*}
	F^\omega_{LL}(\rho) = \inf_{\rho_\psi = \rho} \left\langle \psi, H_N^{0,\omega} \psi \right\rangle = \inf_{\rho_\psi = \rho} \left\langle \psi, \left( \sum_{i=1}^N \Delta_{x_i} + \sum_{1\leq j < k \leq N} \omega(x_j - x_k) \right) \psi \right\rangle.
\end{equation*}
Here, one has to think about for which functions to define this functional. In this case, it is natural to assume $\sqrt{\rho} \in H^1(\Rbb^d)$ and $\int_{\Rbb^d} \rho = N$. This follows from the Hofmann-Ostenhoff inequality, which we will see in Section \ref{chapter:known_bds}, and the representability of such densities, see \cite{Lieb_1983}. Now (\ref{eq:dft}) reads
\begin{equation*}
	E_N^{V,\omega} = \inf_\rho \left\lbrace F^\omega_{LL}(\rho) + \int_{\Rbb^d} V(x) \rho(x) \; dx \right\rbrace.
\end{equation*}
The functional $F^\omega_{LL}$ is universal in the sense that it is independent of the external potential $V$. In this case, it depends on the particle number $N$, which is ignored here for simplicity.

There are several extensions to the functional $F^\omega_{LL}$ (for a recent overview, see \cite{Lewin_2019}) which not only deal with pure states $\psi$ but also mixed states or even grand-canonical states. We will give a precise definition of the functional we use in this work in Section \ref{chapter:main_res}.}

\subsection{The Local Density Approximation}\label{sec:lda}
{While there exist several approximations to the energy of a density, we will be interested in the \textit{Local Density Approximation}
\begin{equation*}
	LDA_f(\rho) = \frac{1}{2}\int_{\Rbb^d}\int_{\Rbb^d} \rho(x) \rho(y) \omega(x-y) \; dx \; dy + \int_{\Rbb^d} f(\rho(x)) \; dx
\end{equation*}
where the function $f : \Rbb_+ \to \Rbb$ is supposed to approximate the local part of the energy. The first term is called the \textit{direct term} and encodes the classical interaction energy of the density $\rho$. 

One of the oldest LDA's for Coulomb interaction is the Thomas-Fermi (TF) functional \cite{Thomas_1927,Fermi_1927}, where one takes $f(t) = 3/10 (3\pi^2)^{2/3}t^{5/3}$ in dimension $d=3$. While this functional is rather simple and fails to predict many basic properties of real-world systems, it is surprisingly accurate in some cases \cite{Solovej_2015}.

There are some generalizations of this approximation, the simplest being the Local Spin-Density Approximation (LSDA) which takes into account that real-world electrons have two spin states. Another, even more general ansatz is the Generalized Gradient Approximation (GGA), where not only the values of $\rho$ are taken into account, but also its gradient $\nabla \rho$. For a detailed overview and assessment of some density functionals used in computational chemistry, see \cite{MHG_2017}.

In this work, we will justify the use of $f = \elda$, where $\elda(\rho_0)$ is the energy of the constant density $\rho_0$ per unit volume, for which the classical interaction energy has been dropped. Since constant, non-zero functions cannot be densities of quantum mechanical states in our setting, we will use a suitable thermodynamic limit.

While density functional theory and the LDA are mostly used for Coulomb interaction because it is the physically relevant setting for atoms and molecules, it is still interesting to think about other interaction potentials. The first mathematically rigorous justification for the Thomas-Fermi functional was given by Lieb and Simon \cite{LS_1977} in 1977 for the Coulomb case. In their paper, the authors proved that the ground state energy can be approximated by the respective minimization problem for the Thomas-Fermi functional.

Recently, this result was extended in \cite{FLS_2019} where it was proven that the TF functional gives the correct ground state energy under relatively weak assumptions on the interaction and external potential. While the previous results restrict themselves to the ground state energy, it was proven in \cite{GN_2018} that the Levy-Lieb functional converges in a suitable sense to the TF functional. This was done for a large set of interaction potentials. In turn, this gives an easy proof of the validity of the ground state TF energy.

Most of these works concern themselves with the aforementioned choice of $f(t) = Ct^{5/3}$ and try to justify this ansatz. Another strategy is to take a suitable thermodynamic limit to obtain the function $f$ in the first place. While this approach is more accurate, because one can just define $f$ as it should be, namely if one believes in the LDA to be true, then the function used can be precisely defined by the thermodynamic limit. However, then it is harder to justify the approximation, because usually the resulting function $f$ has no explicit form and will probably be more complicated than the simple function in TF theory.

This strategy of using the thermodynamic limit in LDA has recently been employed by Lewin, Lieb and Seiringer \cite{Lewin_2020} for Coulomb interaction. More precisely, they rigorously derived the LDA in the regime of slowly varying densities. In this work, we will derive similar estimates for smooth, short-range interaction potentials. Those interactions are not necessarily easier to treat. This is because one cannot use the Graf-Schenker inequality \cite{GS_1995}, which makes use of screening effects for the localization of the energy. Another difficulty is that the Couloumb interaction provides some obvious scaling properties for the resulting thermodynamic limit which are not directly available to us. These properties will be important in deriving a form of continuity which is needed to prove the overall validity of the LDA.
}

{\bf Acknowledgement.} I would like to express my deepest thanks to my mentor, Phan Th\`anh Nam, for his great and continued support.

\section{Main results} \label{chapter:main_res}
{
In this Section, we will give the necessary definitions before stating our main result. Additionally, we will give a short overview over the proof.
}

\subsection{Definitions}\label{sec:def}
{We start by defining the types of interactions that we can treat with our analysis before giving a precise definition of the Levy-Lieb functional used here. For simplicity, we will work with spinless fermions. While we expect that one can follow our arguments here to treat particles with spin, it simplifies our notation drastically. Also, from now on, we will restrict ourself to particles in $\Rbb^3$.

\begin{definition}
\label{def:short-range}
	We call a function $\omega \colon \Rbb^3 \to \Rbb$ \textit{short-range interaction potential} if and only if it is radially symmetric, $0 \leq \omega \in L^1, \; 0 \leq \widehat{\omega} \in L^1, \; \intR |x| \omega(x) \; dx < \infty$ and $\omega$ itself and its Fourier Transform are radially decreasing.
\end{definition}

\begin{remark}
	With this definition, we see immediately that there is a constant $C$ such that $\omega \leq C/|\cdot|$. In particular we have $D_\omega(\rho) \leq CD_c(\rho)$ for all $\rho \geq 0$. Here
	\begin{equation*}
		D_\omega(\rho) \coloneqq \frac{1}{2} \intR \intR \rho(x) \rho(y) \omega(x-y) \; dx \; dy
	\end{equation*}
	is the classical interaction energy of a density $\rho$ and $D_c$ is the Coulomb interaction with $\omega = 1/|\cdot|$. For the definition of $\elda$ in Theorem \ref{thm:uniformgas} below, we only need that the direct term of the interaction $\omega$ in bounded by its Coulomb counterpart. However, then the definition is not as general as presented here. Moreover, the decreasing properties are needed for Theorem \ref{thm:lda} where we compare the Levy-Lieb functional to our LDA.
\end{remark}

A priori, $D_\omega(\rho)$ could be infinity, but we will mostly be talking about densities which belong to $L^1 \cap L^2$ which implies that the interaction energy is finite for short-range potentials and also for the Coulomb interaction by the Hardy-Littlewood-Sobolev inequality. We will also use the notation
\begin{equation}
	D_\omega(\rho_1,\rho_2) \coloneqq \frac{1}{2} \intR \intR \rho_1(x) \rho_2(y) \omega(x-y) \; dx \; dy
\end{equation}
for the interaction energy between two densities $\rho_1$ and $\rho_2$.

As we later want to define the minimal energy of all quantum mechanical states having a particular density, we first define what we mean by a state.

\begin{definition}
\label{def:state}
	A grand-canonical state $\Gamma$ (commuting with the particle number operator) is a collection $\Gamma = \bigoplus_{n \geq 0} \Gamma_n$ of non-negative self-adjoint trace class operators, each acting on $\mathfrak{H}^n$, where
	\begin{equation*}
		\mathfrak{H}^n \coloneqq L^2_a(\Rbb^{3n},\Cbb)
	\end{equation*}
	is the $n$-particle space of antisymmetric square-integrable functions on $\Rbb^{3n}$. Furthermore
	\begin{equation}
		\Gamma_0 + \sum_{n \geq 1} \Tr_{\mathfrak{H}^n} \Gamma_n = 1
	\end{equation}
	where $\Gamma_0 \in [0,1]$ is the probability that there is no particle at all.
\end{definition}

Now we are ready to define the kinetic and interaction energies of a grand-canonical state:

\begin{definition}
\label{def:energies}
	For a grand-canonical state $\Gamma = \bigoplus_{n \geq 0} \Gamma_n$ on Fock space (commuting with the particle number operator) and an interaction potential $\omega$, we introduce the notation
	\begin{equation}
		\Tcal(\Gamma) \coloneqq \sum_{n \geq 1} \Tr_{\mathfrak{H}^n} \left( - \sum_{j=1}^n \Delta_{x_i} \right) \Gamma_n
	\end{equation}
	for the kinetic energy and
	\begin{equation}
		\Ccal_\omega(\Gamma) \coloneqq \sum_{n \geq 2} \Tr_{\mathfrak{H}^n} \left( \sum_{1 \leq j < k \leq n} \omega(x_j - x_k) \right) \Gamma_n
	\end{equation}
	for the interaction energy of the state $\Gamma$.
\end{definition}
If it is clear which interaction potential is used, we will just write $\Ccal(\Gamma)$. We define the density of $\Gamma = \bigoplus_{n \geq 0} \Gamma_n$ by defining the density for each $\Gamma_n$ as follows:
\begin{equation}
	\rho_{\Gamma_n}(x) = n \times \int_{\Rbb^{3(n-1)}} \Gamma_n(x,x_2,\dots,x_n;x,x_2,\dots,x_n) \; dx_2 \cdots dx_n
\end{equation}
where $\Gamma_n(x_1,\dots,x_n;x_1',\dots,x_n')$ is the kernel of the trace-class operator $\Gamma_n$. Now the density $\rho_\Gamma$ is defined by  $\rho_\Gamma = \sum_{n=1}^\infty \rho_{\Gamma_n}$. Sometimes, we will only discuss the lowest possible kinetic energy, which is given by
\begin{equation}
	T(\rho) \coloneqq \inf_{\substack{\Gamma_n = \Gamma_n^* \geq 0 \\ \sum_{n=0}^\infty \Tr(\Gamma_n) = 1 \\ \sum_{n=1}^\infty \rho_{\Gamma_n} = \rho}} \left\lbrace \sum_{n=1}^\infty \Tr_{\mathfrak{H}^n} \left( - \sum_{j=1}^n \Delta_{x_j} \right) \Gamma_n \right\rbrace  = \inf_{\substack{0 \leq \gamma = \gamma^* \leq 1 \\ \rho_\gamma = \rho}} \Tr (-\Delta)\gamma.
\end{equation}
Finally, we define the \textit{grand-canonical Levy-Lieb functional}:
\begin{definition}
\label{def:FLL}
For a density $\rho \in L^1(\Rbb^3,\Rbb_+)$ such that $\sqrt{\rho} \in H^1(\Rbb^3)$ and an interaction potential $\omega$, define
\begin{equation}
\label{eq:FLL}
\begin{split}
	F_{\mathrm{LL}}^\omega(\rho) & \coloneqq \inf_{\substack{\Gamma state\\ \rho_\Gamma = \rho}} \Tcal(\Gamma) + \Ccal_\omega(\Gamma) \\
	& \coloneqq \inf_{\substack{\Gamma_n = \Gamma_n^* \geq 0 \\ \sum_{n=0}^\infty \Tr(\Gamma_n) = 1 \\ \sum_{n=1}^\infty \rho_{\Gamma_n} = \rho}} \left\lbrace \sum_{n=1}^\infty \Tr_{\mathfrak{H}^n} \left( - \sum_{j=1}^n \Delta_{x_j} + \sum_{1 \leq j < k \leq n} \omega(x_j - x_k) \right) \Gamma_n \right\rbrace.
\end{split}
\end{equation}
	For our analysis, it is useful to subtract the direct term from $F_{\mathrm{LL}}$, hence to define
	\begin{equation}
		E^\omega(\rho) \coloneqq F_{\mathrm{LL}}^\omega(\rho) - \frac{1}{2} \intR \intR \rho(x) \rho(y) \omega(x-y) \; dx \; dy.
	\end{equation}
	Again, $\omega$ is dropped from the notation if it is clear which $\omega$ is meant.
\end{definition}

In our analysis, we will be using the grand-canonical version to be able to use states with arbitrary number of particles. For a detailed overview of different types of the Levy-Lieb functional, see for example \cite{Lewin_2019}. In their paper, the authors also prove the existence of optimal states in (\ref{eq:FLL}) which we can use here, since our interaction potential will be positive. Therefore, we do not need to think about minimizing sequences, which will again simplify our arguments.}

\subsection{Main results}\label{sec:statement}
{We will prove the following existence Theorem which we need for our Local Density Approximation.
\begin{main_theorem}[Existence of the Local Density functional]
\label{thm:uniformgas}
	Let $\omega$ be a short-range interaction potential. Let $\rho_0 > 0$ and $\{\Omega_N\} \subset \Rbb^3$ be a sequence of bounded connected domains with $|\Omega_N| \to \infty$, such that $\Omega_N$ has a uniformly regular boundary in the sense that
	\begin{equation*}
		|\partial \Omega_N + B_r| \leq Cr|\Omega|^{2/3}, \quad for \; all \; r \leq |\Omega|^{1/3}/C,		
	\end{equation*}
	for some constant $C > 0$. Let $\delta_N$ be any sequence such that $\delta_N/|\Omega_N|^{1/3} \to 0$ and $\delta_N|\Omega_N|^{1/3} \to \infty$. Let $\chi \in L^1(\Rbb^3)$ be a radial non-negative function of compact support such that $\intR \chi = 1$ and $\intR |\nabla \sqrt{\chi} |^2 < \infty$. Denote $\chi_\delta(x) = \delta^{-3} \chi(x/\delta)$. Then the following thermodynamic limit exists
	\begin{equation}
		\lim_{N \to \infty} \frac{E^\omega(\rho_0\mathds{1}_{\Omega_N} \ast \chi_{\delta_N})}{|\Omega_N|} = \elda^\omega(\rho_0)
	\end{equation}
	where the function $\elda^\omega$ is independent of the sequence $\{\Omega_N\}$, of $\delta_N$ and of $\chi$.
\end{main_theorem}

Now we are ready to state our main result.
\begin{main_theorem}[Validity of the Local Density Approximation]
\label{thm:lda}
	Let $\omega$ be a short-range interaction potential. Let $p > 3$ and $0 < \theta < 1$ such that
	\begin{equation}
		2 \leq p\theta \leq 1 + \frac{2}{5}p.
	\end{equation}
	Then there exists a constant $C = C(p,\theta,\omega)$ such that
	\begin{multline}
	\label{eq:lda}
		\left| E^\omega(\rho) - \intR \elda^\omega(\rho(x))dx \right| \leq \eps \intR (\rho(x) + \rho(x)^2)dx \\
		+ \frac{C}{\eps} \intR |\nabla \sqrt{\rho}(x)|^2 dx +\frac{C}{\eps^{5/2p-1}} \intR |\nabla \rho^\theta (x)|^p dx
	\end{multline}
	for every $\eps > 0$ and every non negative density $\rho \in L^1(\Rbb^3) \cap L^2(\Rbb^3)$ such that $\nabla \sqrt{\rho} \in L^2(\Rbb^3)$ and $\nabla \rho^\theta \in L^p(\Rbb^3)$.	
\end{main_theorem}

When we optimize over $\eps$, this bound yields good results when the gradient terms are much smaller than the local term

\begin{equation*}
\begin{dcases}
	\intR |\nabla \sqrt{\rho}(x)|^2 \; dx \ll \intR \left( \rho(x) + \rho(x)^2 \right) \; dx \\
	\intR |\nabla \rho^\theta(x)|^p \; dx \ll \intR \left( \rho(x) + \rho(x)^2 \right) \; dx 
\end{dcases}
\end{equation*}

In particular, when we take 
\begin{equation*}
	\rho_N(x) = \rho(N^{-1/3}x)
\end{equation*}
and $\eps = N^{-2/15}$, we obtain with $p = 4$, $\theta = 1/2$
\begin{multline}
\label{eq:main_N}
	\left| F_{\mathrm{LL}}^\omega(\rho_N) - \frac{N^2}{2}\intR\intR \rho(x) \rho(y) \omega(N^{1/3}(x-y))\; dx \; dy - N \intR \elda^\omega(\rho(x)) \; dx \right| \\
	\leq CN^{7/15} \intR |\nabla \sqrt{\rho}(x)|^2 \; dx + CN^{13/15} \intR \left( \rho(x) + \rho(x)^2 + |\nabla \sqrt{\rho} (x)|^4 \right) \; dx
\end{multline}

which justifies the LDA in the limit $F_{\mathrm{LL}}^\omega(\rho_N)/N$.

\begin{main_remark}
	These two theorems are in spirit of recent work by Lewin, Lieb and Seiringer \cite{Lewin_2020}, who proved similar bounds for Coulomb interaction $\omega = 1/|\cdot|$. Then the function $\elda$ is called $e_{\mathrm{UEG}}$, the energy of the uniform electron gas. This object is of mathematical interest itself and has been studied, for example, in \cite{Lewin_2017}. This function could be defined for a larger class of interaction potentials, but then it is not so easy to allow indicator functions of the form which is stated in Theorem \ref{thm:uniformgas}.
\end{main_remark}

\begin{main_remark}
	The main error term in our bound is $\frac{1}{\eps^{5/2p-1}} \intR |\nabla \rho^\theta (x)|^p dx$. It is responsible for the large power of $N$ in equation (\ref{eq:main_N}). However it is slightly better than the factor $1/\eps^{4p-1}$, which then leads to $N^{11/12}$ in equation (\ref{eq:main_N}), which was derived in \cite{Lewin_2020}.	However, it is expected \cite{Kirzhnits_1957, Lewin_2020} that the next order of expansion involves at least a gradient term as correction to the kinetic energy which has order $N^{1/3}$.

	One reason we are very far from this is the subadditivity estimate in Lemma \ref{lem:subadd} where we get an error of the form $\eps^{-1}D(\rho)$ which then turns into the error $\eps^{1-5/2p} \intR |\nabla \rho^\theta|^p$ in our case. Here we see why we could improve a little in comparison to the Coulomb case, we can use the boundedness of $\omega$ to estimate the direct term against the $L^2$ norm, whereas for the Coulomb interaction, one has to use the Hardy-Littlewood-Sobolev inequality with the $L^{6/5}$ norm. We see that even when we assume our interaction to be much more regular, we still only improved our error exponent from $11/12$ to $13/15$. This shows that there are certainly new ideas needed to improve the order of convergence, in particular, this is not a Coulomb specific problem.

	On the other hand, we will see that the regularity of the function $\elda$ is important for the validity of the LDA. While regularity is easy to prove for interactions with scaling properties such as Coulomb, a slightly weaker regularity still holds true if $\omega$ and $\hat{\omega}$ are radially decreasing. This allows us to investigate the local behaviour of $\elda$ without having to compute or finding numerical approximations to it. We require $\omega$, $\hat{\omega} \in L^1$ but it seems interesting to find less restrictive assumptions such that $\elda$ still has enough regularity.
\end{main_remark}
}

\subsection{Method of proof}\label{sec:method}
{We will now explain the general ideas of the proof and how this \type{} is structured. In general, we will follow the arguments from Lewin, Lieb and Seiringer \cite{Lewin_2020} and adapt their proofs when necessary.

At first, in Section \ref{chapter:known_bds} we will review some known bounds on $E(\rho)$ which we will use throughout his work. Then, in Section \ref{chapter:loc} we will explain how to use a partition of unity $\sum_k \chi_k = 1$ to approximate the full energy of a density by sums of energies of cut-off densities
\begin{equation*}
	E(\rho) \approx \sum_k E(\chi_k \rho).
\end{equation*}
This approximation will be useful in the proof of Theorem \ref{thm:lda} as well as in Section \ref{chapter:elda} where we will prove the existence of the function $\elda$. First, we will prove it only in the form
\begin{equation*}
	\elda(\rho_0) = \lim_{\ell \to \infty} \frac{E(\rho_0\indvar{\ell\Delta}{})}{|\ell\Delta|}
\end{equation*}
for a fixed tetrahedron $\Delta$ and a fixed regularizing function $\eta$, but we will bootstrap this argument to prove the general Theorem \ref{thm:uniformgas}. Moreover we will derive quantitative estimates that will justify the approximation
\begin{equation*}
	E(\rho_0\indvar{\ell\Delta}{}) \approx \intR\elda(\rho_0)(\indvar{\ell\Delta}{}).
\end{equation*}
Also in this section we investigate the continuity of $\elda$. Here, the challenge was that in \cite{Lewin_2020}, the scaling properties of the Coulomb interaction allowed to prove scaling estimates on $e_{\mathrm{UEG}}$. Since we do not have those properties for $\omega$, we will discuss how to use another property, namely that $\omega$ is decreasing to work around that problem.

In Section \ref{chapter:constden}, we will see how to relate $E(\chi_k \rho)$ to $E(\chi_k \rho_0)$ for some $\rho_0$ in the support of $\chi_k$. To this end, we will derive subadditivity estimates on the energy which unfortunately are responsible for the large error term in Theorem \ref{thm:lda}. Some of these estimates will also be used together with the continuity of $\elda$ to derive
\begin{equation*}
	\intR\elda(\rho_0)(\indvar{\ell\Delta}{}) \approx \intR\elda(\rho)(\indvar{\ell\Delta}{}).
\end{equation*}
We will put all our estimates together in Section \ref{chapter:mainthm} to prove Theorem \ref{thm:uniformgas}, which then has the form
\begin{equation*}
	E(\rho) \approx \intR\elda(\rho).
\end{equation*}

From here on, $C > 0$ will denote a positive constant. It will sometimes change from line to line, but it only depends on the interaction potential $\omega$ and the two parameters $p$ and $\theta$ which appear in the statement of Theorem \ref{thm:lda}. It should be possible to extract a value for $C$ from our proof, but it is highly non-optimal, so we did not try it.}

\section{Known bounds\texorpdfstring{ on $E(\rho)$}{}}\label{chapter:known_bds}
{
Here, we will present some well-known bounds on $E(\rho)$. It will turn out that for the upper bound, we use only the kinetic part, whereas for the lower bounds, the interaction potential plays an important role.

\begin{lemma}
\label{lem:intomega}
	For any radially symmetric function $\omega$ and any grand-canonical state $\Gamma$ in Fock space, we have
	\begin{align}
		i) &\quad \Ccal_\omega(\Gamma) - D_\omega(\rho_\Gamma) \geq - \frac{\omega(0)}{2} \intR \rho_\Gamma \quad \quad &when \; \widehat{\omega} \in L^1(\Rbb^3) \; and \; \widehat{\omega} \geq 0, \\
		ii) &\quad \Ccal_\omega(\Gamma) - D_\omega(\rho_\Gamma) \geq - \frac{\intR \omega}{2} \intR \left( \rho_\Gamma \right)^2 &when \; \omega \in L^1(\Rbb^3) \; and \; \omega \geq 0.
	\end{align}
\end{lemma}

\begin{proof}
	Inequality $ii)$ follows from $\Ccal_\omega(\Gamma) \geq 0$ and $D_\omega(\rho_\Gamma) \leq \lVert \omega \rVert_{L^1} \lVert \rho_\Gamma \rVert_{L^2}^2/2$ by Young's inequality. For the first inequality, we write as in \cite{Lewin_2017}, for any fixed particle number $N$ and any $\rho$
	\begin{equation*}
	\begin{split}
	0 & \leq \frac{1}{2}	\intR \intR \omega(x-y) \left( \sum_{j=1}^N \delta_{x_j}(x) - \rho(x) \right)  \left( \sum_{j=1}^N \delta_{x_j}(y) - \rho(y) \right) dx \; dy \\
	& = \sum_{1 \leq j < k \leq N} \omega(x_j - x_k) + N\frac{\omega(0)}{2} - 2 \sum_{j=1}^N D_\omega	(\rho,\delta_{x_j}) + D_\omega(\rho,\rho),
	\end{split}
	\end{equation*}
	where the non-negativity of $D_\omega$ follows from $\hat{\omega} \geq 0$. Taking $\rho = \rho_\Gamma$, tracing this against the $N$-particle component of $\Gamma$ and summing up all those components yields $i)$.
\end{proof}

For a lower bound on $E(\rho)$, recall some estimates on the kinetic energy. An important tool is the Lieb-Thirring inequality \cite{LS_2010,LT_1975}, which states that there exists a positive constant $1/C > 0$, depending on the dimension $d$ such that
\begin{equation}
\label{eq:LT}
	\Tr(-\Delta)\gamma \geq \frac{1}{C} \int_{\Rbb^d} \rho_\gamma(x)^{1+\frac{2}{d}} \; dx
\end{equation}
for every self-adjoint operator $\gamma$ on $L^2(\Rbb^d)$ such that $0 \leq \gamma \leq 1$. In dimension $d \geq 3$, the constant $1/C$ has been conjectured to be
\begin{equation}
	\clt = \frac{4\pi^2d}{(d+2)} \left(\frac{d}{|\Sbb^{d-1}|} \right)^{\frac{2}{d}}. 
\end{equation}
This constant can be archieved by adding gradient corrections, which is done in \cite{Nam_2018} where the bound
\begin{equation}
	\Tr(-\Delta)\gamma \geq \clt(1-\eps)\int_{\Rbb^d} \rho_\gamma(x)^{1+\frac{2}{d}} \; dx - \frac{\kappa}{\eps^{3+\frac{4}{d}}} \int_{\Rbb^d} |\nabla \sqrt{\rho_\gamma}(x)|^2 \; dx
\end{equation}
for any $\eps > 0$ and some constant $\kappa = \kappa(d)$ for all space dimensions $d \geq 1$ was proved.
Another important bound is the Hoffmann-Ostenhoff inequality \cite{HOHO_1977}:
\begin{equation}
	\Tr(-\Delta)\gamma \geq \int_{\Rbb^d} |\nabla \sqrt{\rho_\gamma}(x)|^2 \; dx
\end{equation}
which does not require the fermionic constraint $0 \leq \gamma \leq 1$ and imposes that $\sqrt{\rho} \in H^1(\Rbb^d)$.
By Lemma \ref{lem:intomega} and the Lieb-Thirring inequality (\ref{eq:LT}), we immediately see

\begin{lemma}[Lower bound on $E(\rho)$]
\label{lem:lowerrough}
	Let $\omega$ be short-range. Then we have
	\begin{align}
		i) &\quad E(\rho) \geq \frac{1}{C} \intR \rho(x)^{\frac{5}{3}} dx - C \intR \rho(x) dx \\
		\intertext{and}
		ii) &\quad E(\rho) \geq \frac{1}{C} \intR \rho(x)^{\frac{5}{3}} dx - C \intR \rho(x)^2 dx
	\end{align}
	for every $\rho \geq 0$ such that $\sqrt{\rho} \in H^1(\mathbb{R}^3)$.
\end{lemma}

These will be the required lower bounds on $E(\rho)$. In many cases, the negative part is enough for us. For the upper bound, we only bound the kinetic part and then conclude by the positivity of $\omega$.

\begin{lemma}[Upper bound on $E(\rho)$]
\label{lem:upperrough}
	Let $\omega$ be positive. Then there exists a constant $C > 0$ such that
	\begin{equation}
		E(\rho) \leq \clt(1 + C\eps) \intR \rho(x)^{\frac{5}{3}} \; dx + \frac{C(1 + \sqrt{\eps})^2}{\eps} \intR {\left| \nabla \sqrt{\rho}(x) \right|}^2 \; dx
	\end{equation}
	for every $\eps > 0$ and $\rho \geq 0$ such that $\sqrt{\rho} \in H^1(\mathbb{R}^3)$.
\end{lemma}

\begin{proof}
	In \cite[Theorem 3]{Lewin_2020} Lewin, Lieb and Seiringer construct a fermionic one-particle density matrix $\gamma$ (i.e., an operator satisfying $0 \leq \gamma =\gamma^*\leq 1$) such that it has the required density
	\begin{equation*}
		\rho_\gamma(x) = \rho(x)
	\end{equation*}
	and
	\begin{equation*}
		\Tr(-\Delta)\gamma = \clt(1 + C\eps) \intR \rho(x)^{\frac{5}{3}} \; dx + \frac{C(1 + \sqrt{\eps})^2}{\eps} \intR {\left| \nabla \sqrt{\rho}(x) \right|}^2 \; dx.
	\end{equation*}
	This density matrix is representable by a quasi-free state $\Gamma_\gamma$ in Fock space, see for example \cite{BLS_1994}. Then the two-particle density matrix is given by 
	\begin{equation*}
		\Gamma^{(2)}_\gamma(x_1,x_2,y_1,y_2) = \gamma(x_1,y_1)\gamma(x_2,y_2) - \gamma(x_1,x_2)\gamma(y_1,y_2), 
	\end{equation*}
	which implies that its interaction energy with the pair potential $\omega$ is
	\begin{equation*}
		\Ccal(\Gamma_\gamma) = \frac{1}{2} \Tr \left( \omega \Gamma^{(2)}_\gamma \right) =  \frac{1}{2} \int_{\Rbb^3} \int_{\Rbb^3} \omega(x-y) \left( \rho(x) \rho(y) - |\gamma(x,y)|^2 \right) dx \; dy
	\end{equation*}
	where $\omega$ is seen as the multiplication operator by $\omega(x_1-x_2)$ on the two-particle space $L^2(\Rbb^3) \wedge L^2(\Rbb^3)$.
	So, when taking the trial state $\Gamma_\gamma$, we get
	\begin{multline*}
		E(\rho) \leq \clt(1 + C\eps) \intR \rho(x)^{\frac{5}{3}} \; dx + \frac{C(1 + \sqrt{\eps})^2}{\eps} \intR {\left| \nabla \sqrt{\rho}(x) \right|}^2 \; dx \\
		- \frac{1}{2} \int_{\Rbb^3} \int_{\Rbb^3} \omega(x-y)|\gamma(x,y)|^2 dx \; dy.
	\end{multline*}
	We now conclude by the positivity of $\omega$.
\end{proof}
}

\section{Localizing the density}\label{chapter:loc}
{
Our goal in this Section is to relate the total energy of some density $\rho$ to the sum of the energies of localized densities.
}

\subsection{Upper bound}\label{sec:localupper}
{Here, we care about the upper bound. Since in our definition of short-range interactions, we assume our interaction to be bounded by the Coulomb interaction, we can almost directly use the arguments from \cite{Lewin_2020}.

Let $C_1 = (-1/2,1/2)^3$ be the unit cube and notice that it is the union of 24 disjoint identical tetrahedra $\Delta_1, \dots ,\Delta_{24}$ which have all volume 1/24. From this, we obtain a tiling of $\Rbb^3$ with tetrahedra
\begin{equation}
	\Rbb^3 = \bigcup_{z \in \Zbb^3} (\overline{C_1} + z) = \bigcup_{z \in \Zbb^3} \bigcup_{j = 1}^{24} (\overline{\Delta_j} + z)
\end{equation}
and a corresponding partition of unity
\begin{equation}
	\sum_{z \in \Zbb^3} \sum_{j = 1}^{24} \mathds{1}_{\ell\Delta_j}(x - \ell z) \equiv 1 \quad for \; a.e. \; x \in \Rbb^3,
\end{equation}
for any fixed tile-size $\ell > 0$. We denote any $\Delta_j$ as $\Delta_j = \mu_j \Delta$ where $\Delta$ is a reference tetrahedron with $0$ as its center of mass. From now on, $\mu_j = (z_j, R_j) \in C_1 \times SO(3)$ is the translation and rotation such that $\Delta_j = R_j \Delta - z_j$. Now define the characteristic function which will be used in our upper bound:
\begin{equation}\label{eq:unityupper}
	\chi_{\ell,\delta,j} \coloneqq \frac{1}{(1-\delta /\ell)^3}\mathds{1}_{\ell\mu_j(1-\delta /\ell)\Delta} \ast \eta_\delta.
\end{equation}
Here $\eta_\delta(x) = (10/\delta)^3\eta_1(10x/\delta)$ where $\eta_1$ is a fixed $C_c^\infty$ non-negative radial function with support in the unit ball and $\intR \eta_1 = 1$. When $\delta \leq \ell/2$, our function $\chi_{\ell,\delta,j}$ has its support inside $\ell \Delta_j$, at a distance proportional to $\delta$ from its boundary. Note that
\begin{equation*}
	\intR \chi_{\ell,\delta,j} = \left| \ell \Delta \right| = \frac{\ell^3}{24}.
\end{equation*}
Hence the function
\begin{equation*}
	\sum_{z \in \Zbb^3} \sum_{j = 1}^{24} \chi_{\ell,\delta,j} (\cdot - \ell z)
\end{equation*}
is our incomplete partition of unity in the sense that it is $(1-\delta/\ell)^{-3} > 1$ inside the tiles but vanishes in a neighbourhood of the boundary of the tiles. We obtain a full partition of unity after averaging over the translations of the tiling:
\begin{equation}
\label{eq:incomp_unity}
	\frac{1}{\ell^3} \int_{C_\ell} \sum_{z \in \Zbb^3} \sum_{j = 1}^{24} \chi_{\ell,\delta,j} (x - \ell z - \tau) d\tau = 1 \quad \textrm{for} \; a.e. \; x \in \Rbb^3.
\end{equation}
Here and subsequently $C_\ell = (-\ell/2,\ell/2)^3 = \ell C_1$ is the cube of side length $\ell$. We see this by noticing that for any $f \in L^1(\Rbb^3)$
\begin{equation*}
	\begin{split}
	\int_{C_\ell} \sum_{z \in \Zbb^3} f(x - \ell z - \tau)d\tau & = \int_{C_\ell} \sum_{z \in \ell \Zbb^3} f(x - z - \tau)d\tau \\
	& = \intR f(x - \tau)d\tau = \intR f(\tau) d\tau.
\end{split}
\end{equation*}

The main result in this section is the following upper bound which is taken almost directly from \cite{Lewin_2020}.
\begin{proposition}[Upper bound in terms of local densities]
\label{prop:energylocalupper}
	Let $\omega$ be an interaction potential such that $D_\omega(\rho) \leq C D_c(\rho)$ for a constant $C$ and any positive $\rho$. Then there exists a universal constant $C$ such that for any $\sqrt{\rho} \in H^1(\Rbb^3)$, any $0 < \delta < \ell/2$ and any $0 < \alpha < 1/2$,
	\begin{multline}
		E(\rho) \leq \left(\int_{1-\alpha}^{1+\alpha} \frac{ds}{s^4} \right)^{-1} \int_{1-\alpha}^{1+\alpha} \frac{dt}{t^4} \int_{SO(3)} dR \int_{C_{t\ell}} \frac{d\tau}{(t\ell)^3} \times \\
		\times \sum_{z \in \Zbb^3} \sum_{j = 1}^{24} E \left( \chi_{t\ell,t\delta,j}(R \cdot - t\ell z - \tau)\rho \right) +  C \delta^2 \log(\alpha^{-1}) \intR \rho^2.
	\end{multline}
	In particular, we can find $\ell' \in (\ell (1- \alpha), \ell (1+ \alpha))$, $\delta' \in (\delta(1- \alpha), \delta(1+ \alpha))$ and an isometry $(\tau, R) \in \Rbb^3 \times SO(3)$ such that
	\begin{equation}
		\sum_{z \in \Zbb^3} \sum_{j = 1}^{24} E \left( \chi_{\ell',\delta',j}(R \cdot - \ell' z - \tau)\rho \right) +  C \delta^2 \log(\alpha^{-1}) \intR \rho^2.
	\end{equation}
\end{proposition}

\begin{proof}
	The proof follows directly \cite{Lewin_2020} in the sense that we repeat the proof with our interaction energy $D_\omega$. Using (\ref{eq:incomp_unity}), we write
	\begin{equation*}
		\rho(x) = \frac{1}{\ell^3} \int_{C_\ell} \sum_{z \in \Zbb^3} \sum_{j = 1}^{24} \chi_{\ell,\delta,j}(x-\ell z - \tau)\rho(x) \; d\tau.
	\end{equation*}
	Now let $\Gamma_{\tau,\ell,\delta} = \bigotimes_{z \in \Zbb^3} \bigotimes_{j=1}^{24} \Gamma_{\tau,\ell,\delta,z,j}$ for any $\tau \in C_\ell$, where $\Gamma_{\tau,\ell,\delta,z,j}$ is the minimizer for the Levy-Lieb energy with density
	\begin{equation*}
		\rho_{\Gamma_{\tau,\ell,\delta,z,j}}(x) = \chi_{\ell,\delta,j}(x-\ell z - \tau)\rho(x).
	\end{equation*}
	Because the states $\Gamma_{\tau,\ell,\delta,z,j}$ have disjoint support, we can anti-symmetrize the state $\Gamma_{\tau,\ell,\delta}$ in a standard manner. Denote this anti-symmetrized state with a slight abuse of notation again by $\Gamma_{\tau,\ell,\delta}$. Its energy and density are the same after anti-symmetrizing and we obtain
	\begin{equation*}
		\rho_{\Gamma_{\tau,\ell,\delta}}(x) = \sum_{z \in \Zbb^3} \sum_{j = 1}^{24} \chi_{\ell,\delta,j}(x-\ell z - \tau)\rho(x).
	\end{equation*}
	Now take the trial state
	\begin{equation*}
		\Gamma_{\ell,\delta} = \frac{1}{\ell^3} \int_{C_\ell} \Gamma_{\tau,\ell,\delta} \; d\tau
	\end{equation*}
	and notice that $\rho_{\Gamma_{\ell,\delta}} = \ell^{-3} \int_{C_\ell} \rho_{\Gamma_{\tau,\ell,\delta}} = \rho$. We get
	\begin{equation}
	\begin{split}
		E(\rho) & \leq \frac{1}{\ell^3} \int_{C_\ell} \sum_{z \in \Zbb^3} \sum_{j = 1}^{24} E \left( \chi_{\ell,\delta,j}(\cdot - \ell z - \tau)\rho \right) \; d\tau + \frac{1}{\ell^3} \int_{C_\ell} D_\omega(\rho_{\Gamma_{\tau,\ell,\delta}}) \; d\tau - D_\omega(\rho) \\
		& = \frac{1}{\ell^3} \int_{C_\ell} \sum_{z \in \Zbb^3} \sum_{j = 1}^{24} E \left( \chi_{\ell,\delta,j}(\cdot - \ell z - \tau)\rho \right) \; d\tau + \frac{1}{\ell^3} \int_{C_\ell} D_\omega(\rho_{\Gamma_{\tau,\ell,\delta}} - \rho) \; d\tau,
	\end{split}
	\end{equation}
	where we have used that the energy with the direct term subtracted of a tensor product of states of disjoint support is the sum of the energies of the pieces \cite{Lewin_2017}. This is actually always true for the interaction, for the kinetic energy we need once more that our states have disjoint support. Next, we average everything over rotations of the tiling and replace $\ell$ by $t\ell$ and $\delta$ by $t\delta$ before averaging over $t$, but most importantly we use that we can estimate our interaction energy against the Coulomb energy to get
	\begin{equation}
	\begin{split}
		& \left(\int_{1-\alpha}^{1+\alpha} \frac{ds}{s^4} \right)^{-1} \int_{1-\alpha}^{1+\alpha} \frac{dt}{t^4} \int_{SO(3)} dR \frac{1}{(t\ell)^3} \int_{C_{t\ell}} D_\omega(\rho_{\Gamma_{\tau,t\ell,t\delta,R}} - \rho) \; d\tau \\
		\leq & \, \left(\int_{1-\alpha}^{1+\alpha} \frac{ds}{s^4} \right)^{-1} \int_{1-\alpha}^{1+\alpha} \frac{dt}{t^4} \int_{SO(3)} dR \frac{1}{(t\ell)^3} \int_{C_{t\ell}} D_c(\rho_{\Gamma_{\tau,t\ell,t\delta,R}} - \rho) \; d\tau \\
		\leq & \, C \delta^2\log(\alpha^{-1}) \intR \rho^2
	\end{split}
	\end{equation}
	where the last estimate is directly from \cite{Lewin_2020}.
\end{proof}
}

\subsection{Lower bound}\label{sec:locallower}
{For the lower bound, we introduce, similarly to (\ref{eq:unityupper})
\begin{equation}
	\xi_{\ell, \delta, j} \coloneqq \mathds{1}_{\ell \mu_j \Delta} \ast \eta_\delta
\end{equation}
which forms a smooth partition of unity without holes
\begin{equation*}
	\sum_{z \in \Zbb^3} \sum_{j = 1}^{24} \xi_{\ell, \delta, j}(x - \ell z) = 1.
\end{equation*}
Recall that $\eta_1$ is a fixed $C_c^\infty$ non-negative radial function with support in the unit ball such that $\intR \eta_1 = 1$ and $\eta_\delta(x) = (10/\delta)^3\eta_1(10x/\delta)$.

\begin{proposition}
\label{prop:energylocallower}
	Let $\omega$ be a positive interaction such that $\omega \in L^1$ and $\intR |x|\omega(x) \; dx < \infty$. Then there exists a universal constant $C$ such that for all $\sqrt{\rho} \in H^1(\Rbb^3)$, tetrahedra $\Delta$ of size $1/24$ and $0 < \delta/\ell < 1/C$
	\begin{equation}
		E(\rho) \geq \frac{1}{\ell^3 |\Delta|} \int_{SO(3)} \intR E \left( \rho(\indld)(R \cdot - z) \right) dz \, dR - \frac{C}{\ell \delta}\intR \rho_\Gamma - C\frac{1+\delta}{\ell} \intR \rho_\Gamma^2.
	\end{equation}
\end{proposition}

\begin{proof}
	Since we have a fixed tetrahedron $\Delta$, let $\xi_{\ell, \delta} \coloneqq \indld$. We start with the kinetic energy. By the IMS formula, on $\Rbb^{3N}$, we have as in \cite{Lewin_2017}
	\begin{multline*}
		-\sum_{i=1}^N \Delta_{x_i} = \frac{1}{|\ell\Delta|} \int_{SO(3)} \intR \sqrt{\xi_{\ell, \delta}(R \cdot - z)} \left( -\sum_{i=1}^N \Delta_{x_i} \right) \sqrt{\xi_{\ell, \delta}(R \cdot - z)} \\
		- \frac{N}{|\ell\Delta|} \intR |\nabla \sqrt{\xi_{\ell, \delta}(R \cdot - z)}|^2
	\end{multline*}
	(see also \cite[Eq. (30)]{Hainzel_2009_2}). Using
	\begin{equation*}
		\frac{1}{|\ell\Delta|} \intR |\nabla \sqrt{\mathds{1}_{\ell\Delta} \ast \eta_\delta (R \cdot - z)}|^2 \leq \frac{C}{\ell\delta},
	\end{equation*}
	we obtain
	\begin{equation}
	\label{eq:energylocallower1}
		\Tcal(\Gamma) \geq \frac{1}{|\ell\Delta|} \int_{SO(3)} \intR \Tcal(\Gamma_{|\sqrt{\xi_{\ell, \delta}(R \cdot - z)}}) - \frac{C}{\ell\delta} \intR \rho_\Gamma,
	\end{equation}
	where we use the notion of quantum localized states $\Gamma_{|f}$, that is the unique state which has the $k$-particle reduced density matrices $f^{\otimes k}\Gamma^{(k)}f^{\otimes k}$. This can be recalled, for example, in \cite[Appendix A.1.2]{Hainzel_2009_2} or in \cite[Section 3]{Lewin_2011} for the truncated Fock space. By definition, the density of the localized state is $\rho_{\Gamma_{|\sqrt{\mathds{1}_{\ell\Delta} \ast \eta_\delta (R \cdot - z)}}} = \rho_{\Gamma}\mathds{1}_{\ell\Delta} \ast \eta_\delta (R \cdot - z)$. Now we prove a similar lower bound for the interaction potential. It will turn out that we can still use a variant of the Graf-Schenker inequality \cite{GS_1995}, even though we don't have a Coulomb potential. For any $\ell$ and $\delta$, define the function $\hti_{\ell,\delta}$ by
	\begin{equation*}
		\hti_{\ell,\delta}(x) \coloneqq \frac{1}{|\ell\Delta|} \int_{SO(3)} \left( \indld \right) \ast \left( \indvar{-\ell\Delta}{\delta} \right) (Rx) \; dR.
	\end{equation*}
	Notice that
	\begin{equation*}
		\hti_{\ell,\delta}(x-y) = \frac{1}{|\ell\Delta|} \int_{SO(3)} \intR (\indvar{R^{-1}\ell\Delta + z}{\delta})(y)(\indvar{R^{-1}\ell\Delta + z}{\delta})(x) \; dz \; dR.
	\end{equation*}
	Since $0\leq \hti_{\ell,\delta} \leq 1$ this implies that for $\widetilde{\omega}_{\ell,\delta}(x) \coloneqq \omega(x) - \hti_{\ell,\delta}(x)\omega(x) \geq 0$ we have the bound, again as in \cite{Lewin_2020}
	\begin{multline}
	\label{eq:energylocallower2}
		\Ccal_\omega(\Gamma) - D_\omega(\rho_\Gamma) - \frac{1}{|\ell\Delta|} \int_{SO(3)} \intR \left\lbrace \Ccal_\omega \left( \Gamma_{|\sqrt{\xi_{\ell, \delta}(R \cdot - z)}} \right) - D_\omega(\rho_{\Gamma}\xi_{\ell, \delta}(R \cdot - z) \right\rbrace dz \; dR \\
		= \Ccal_{\widetilde{\omega}_{\ell,\delta}}(\Gamma) - D_{\widetilde{\omega}_{\ell,\delta}}(\rho_\Gamma) \geq -\frac{1}{2}\intR \widetilde{\omega}_{\ell,\delta} \intR \rho_\Gamma^2.
	\end{multline}
	For the last inequality, we used Lemma \ref{lem:intomega} $ii)$. We claim that under our assumption on $\omega$, there is a constant $C$, which is independent of $\ell$ and $\delta$, such that
	\begin{equation}
	\label{eq:energylocallower3}
		\intR \widetilde{\omega}_{\ell,\delta} \leq \frac{C(1+\delta)}{\ell}.
	\end{equation}
	If we prove this, then equations (\ref{eq:energylocallower1}) and (\ref{eq:energylocallower2}) together will yield our Proposition. To see (\ref{eq:energylocallower3}), we write
	\begin{multline*}
		\intR \widetilde{\omega}_{\ell,\delta}(x) \; dx = \intR (1 - \hti_{\ell,\delta}(x))\omega(x) \; dx \\
		\leq \intR |1-\hti_{\ell,\delta}(0)|\omega(x) \; dx + \intR |\hti_{\ell,\delta}(0) - \hti_{\ell,\delta}(x)|\omega(x) \; dx
	\end{multline*}
	and claim that the following formulae hold:
	\begin{equation*}
	\begin{split}
		i) &\quad \hti_{\ell,\delta}(x) = \hti_{1,\delta/\ell}(x/\ell) \\
		ii)  &\quad |\hti_{1,\delta/\ell}(0) - \hti_{1,\delta/\ell}(x)| \leq C|x| \\
		iii) &\quad |1 - \hti_{1,\delta/\ell}(0)| \leq C \delta/\ell.
	\end{split}
	\end{equation*}
	Formula $i)$ is just a routine computation, which we omit. For formula $iii)$, we write
	\begin{equation*}
		\hti_{1,\delta/\ell}(0) = \frac{1}{|\Delta|} \intR \left( (\indvar{\Delta}{\delta/\ell})(y) \right)^2 \; dy
	\end{equation*}
	and
	\begin{equation*}
		1 = \frac{1}{|\Delta|} \intR \mathds{1}_\Delta(y) \; dy.
	\end{equation*}
	Now, for all $y \in \Rbb^3$ for which $\dist(y,\partial \Delta) > C\delta/\ell$, we have $\left( (\indvar{\Delta}{\delta/\ell})(y) \right)^2 = \mathds{1}_\Delta(y)$. In this context, the constant $C$ only depends on the chosen function $\eta$. The set where this fails is a set of order $\delta/\ell$ since we always have $\delta/\ell$ small. So we obtain
	\begin{equation*}
	\begin{split}
		|1 - \hti_{1,\delta/\ell}(0)| & \leq \frac{1}{|\Delta|} \int_{\dist(y,\partial \Delta) \leq C\delta/\ell} \left| \left( (\indvar{\Delta}{\delta/\ell})(y) \right)^2 - \mathds{1}_\Delta(y) \right| \; dy \\
		& \leq C \int_{\dist(y,\partial \Delta) \leq C\delta/\ell} dy \\
		& \leq C\frac{\delta}{\ell}.
	\end{split}
	\end{equation*}
	For formula $ii)$, we claim that $\hti_{1,\delta/\ell}$ is Lipschitz with a Lipschitz constant independent of $\ell$ and $\delta$. To see that, first notice that the function $\mathds{1}_\Delta \ast \mathds{1}_{-\Delta}$ is clearly Lipschitz. Now we argue that everything we are doing to generate $\hti$  from this function does not change the Lipschitz constant. Indeed, if $f$ is Lipschitz and $g \in L^1$ is non-negative, we have:
	\begin{equation*}
		|(g \ast f)(x) - (g \ast f)(y)| \leq \int g(z)|f(x-z)-f(y-z)| \; dz \leq C \left( \int g(z) \; dz \right) |x-y|.
	\end{equation*}
	Here, we use $g = \eta_{\delta/\ell} \ast \eta_{\delta/\ell}$ for which $\int g = 1$ holds independently of $\ell$ and $\delta$. Therefore, $\left( \indvar{\Delta}{\delta/\ell} \right) \ast \left( \indvar{-\Delta}{\delta/\ell} \right)$ is Lipschitz and integrating over $SO(3)$ leaves this invariant. This concludes $ii)$ and thereby the whole proof.
\end{proof}

\begin{corollary}[Lower bound in terms of local densities]
\label{cor:energylocallower}
	Let $\omega$ be a positive interaction such that $\omega \in L^1$ and $\intR |x|\omega(x) \; dx < \infty$. Then there exists a constant $C$ such that for any $\sqrt{\rho} \in H^1(\Rbb^3)$ and any $0 < \delta/\ell < 1/C$
	\begin{equation}
		E(\rho) \geq \frac{1}{\ell^3} \sum_{z \in \Zbb^3} \sum_{j = 1}^{24} \int_{SO(3)} \int_{C_\ell} E \left( \xi_{\ell, \delta, j}(R \cdot - \ell z - \tau) \rho \right) d\tau \, dR - \frac{C}{\ell \delta} \intR \rho - \frac{C(1+\delta)}{\ell} \intR \rho^2.
	\end{equation}
	In particular, we can find an isometry $(\tau, R) \in \Rbb^3 \times SO(3)$ such that
	\begin{equation}
		E(\rho) \geq \sum_{z \in \Zbb^3} \sum_{j = 1}^{24} E \left( \xi_{\ell, \delta, j}(R \cdot - \ell z - \tau) \rho \right) - \frac{C}{\ell \delta} \intR \rho - \frac{C(1+\delta)}{\ell} \intR \rho^2.
	\end{equation}
\end{corollary}

\begin{proof}
	This follows immediately from Proposition \ref{prop:energylocallower} and the fact that for any $f \in L^1$
	\begin{equation*}
		\intR f(z) \, dz = \int_{C_\ell} \sum_{z \in \Zbb^3} f(\ell z + \tau) \, d\tau.
	\end{equation*}
\end{proof}
}

\section{The LDA in the thermodynamic limit}\label{chapter:elda}
{
In this Section, we define the function $\elda$ which is needed to formulate our density functional. We will derive some properties and estimates for which we need different assumptions on $\omega$. For simplicity, from now on $\omega$ will be a fixed short-range interaction potential in the sense of Definition \ref{def:short-range}, but some results hold true with weaker assumptions as we will see. Most of those properties will be used to prove Theorem \ref{thm:uniformgas}, which is also done in this Section.

\begin{definition}
\label{def:energyvolume}
	For $0 < \ell$, $\delta$ we define the energy per unit volume of a tetrahedron at constant density $\rho_0$ by
	\begin{equation}
		e_\Delta(\rho_0,\ell,\delta) \coloneqq \left| \ell \Delta \right|^{-1} E \left( \rho_0 \indld \right),
	\end{equation}
	where $\eta_\delta(x) = (10/\delta)^3\eta(10x/\delta)$ with $\eta$ a fixed $C_c^\infty$ non-negative radial function with support in the unit ball and such that $\intR \eta = 1$.
\end{definition}

Now, we define the energy of the free Fermi gas in a natural way, namely as the energy we would get if the density is constant. As a constant, non-zero function is not in $L^2(\Rbb^3)$, we take the thermodynamic limit of a smeared out characteristic function.

\begin{lemma}
\label{lem:euegex}
	Let $\Delta$ be any tetrahedron in $\Rbb^3$ and $\eta \in C_c^\infty(\Rbb^3, \Rbb^+)$ be a radial function with $\intR \eta = 1$. For any $\rho_0 > 0$, the limits 
	\begin{equation}
		\lim_{\substack{L_N \to \infty \\ \frac{\ell_N}{L_N} \to 0}} \left( \inf_{\substack{\sqrt{\rho} \in H^1(\Rbb^3) \\ \rho_0\mathds{1}_{(L_N - \ell_N)\Delta} \leq \rho \leq \rho_0\mathds{1}_{(L_N + \ell_N)\Delta}}} \frac{E(\rho)}{(L_N)^3 |\Delta|} \right) = \lim_{\ell \to \infty} \frac{E(\rho_0 \mathds{1}_{\ell\Delta} \ast \eta)}{\ell^3 |\Delta|} \eqqcolon \elda(\rho_0)
	\end{equation}
	exist, coincide and are independent of the tetrahedron $\Delta$, $\eta$ and the sequences $L_N$, $\ell$, $\ell_N \to \infty$.
\end{lemma}

\begin{proof}
	For shortness of notation, define
	\begin{equation*}
		u_N = \inf_{\substack{\sqrt{\rho} \in H^1(\Rbb^3) \\ \rho_0\mathds{1}_{(L_N - \ell_N)\Delta} \leq \rho \leq \rho_0\mathds{1}_{(L_N + \ell_N)\Delta}}} \frac{E(\rho)}{(L_N)^3 |\Delta|}
	\end{equation*}
	and
	\begin{equation*}
		v(\ell) = \frac{E(\rho_0 \mathds{1}_{\ell\Delta} \ast \eta)}{\ell^3 |\Delta|}.
	\end{equation*}
	Because we have $\eta$ fixed, we see immediately that $u_N \leq v(L_N)$ for $N$ large enough. Also because $\eta$ is fixed, for all densities $\rho$ satisfying the constraints in the definition of $u_N$ the inequality from Proposition \ref{prop:energylocallower} takes the form
	\begin{equation*}
		E(\rho) \geq \frac{1}{|\ell \Delta|} \int_{SO(3)} \intR E(\rho (\mathds{1}_{\ell\Delta} \ast \eta)(R\cdot -z)) dz \; dR - \frac{c}{\ell} \intR \rho_\Gamma - \frac{c}{\ell} \intR \rho_\Gamma^2.
	\end{equation*}
	Now take $\ell \ll L_N$. Then the measure of all $R$ and $z$ such that $\rho = \rho_0$ on the support of $\mathds{1}_{\ell\Delta} \ast \eta (R \cdot - z)$ has order $|L_N \Delta|$ but it is always smaller. Denote the set where this happens by $\Omega$. If this is the case, we can say
	\begin{equation}
		E(\rho (\mathds{1}_{\ell\Delta} \ast \eta)(R \cdot - z)) = E(\rho_0 (\mathds{1}_{\ell\Delta} \ast \eta)(R \cdot - z)) = E(\rho_0 \mathds{1}_{\ell\Delta} \ast \eta) = v(\ell)|\ell \Delta|
	\end{equation}
	by the rotation and translation invariance of $E$. When $\rho_\Gamma = 0$ on the support of $\mathds{1}_{\ell\Delta} \ast \eta (R \cdot - z)$, the energy is zero. In the case where $\rho_\Gamma$ is not one or zero, we use our lower bound from Lemma \ref{lem:lowerrough} to get $E(\rho (\mathds{1}_{\ell\Delta} \ast \eta)(R \cdot - z)) \geq -C\rho_0 |\ell \Delta|$. This happens in a set which has distance of order $\ell + \ell_N$ to the boundary of $L_N \Delta$, so this set has measure bounded by $C(\ell + \ell_N)(L_N)^2$.
	In total, we get
	\begin{equation*}
		\frac{E(\rho)}{(L_N)^3|\Delta|} \geq \frac{|\Omega|}{|L_N \Delta|} v(\ell) - C\rho_0 \frac{\ell + \ell_N}{L_N} - C\rho_0 \frac{(L_N + \ell_N)^3}{\ell (L_N)^3} - C\rho_0^2 \frac{(L_N + \ell_N)^3}{\ell (L_N)^3} 
	\end{equation*}
	Now there are two cases: if $v(\ell)$ is negative, we use $|\Omega| \leq |L_N \Delta|$, if it is positive, we use $|\Omega| \geq |L_N \Delta| - C(\ell + \ell_N)(L_N)^2$ again, because the set where something goes wrong is a $C(\ell + \ell_N)$-neighbourhood of the boundary $L_N \Delta$. In both cases, after minimizing over $\rho$
	\begin{equation*}
		\liminf_{N \to \infty} u_N \geq v(\ell) - C \frac{\rho_0 + \rho_0^2}{\ell}
	\end{equation*}
	and thus
	\begin{equation*}
		\liminf_{N \to \infty} u_N \geq \limsup_{\ell \to \infty} v(\ell).
	\end{equation*}
	Since $u_N \leq v(L_N)$ we see that both sequences have the same limit independent of $\ell$ and $L_N$.
	
	Furthermore, $\eta$ does not appear in the definition on the left, so it doesn't change $\elda$. Additionally, $\ell_N$ does not appear on the right side. Obviously, our definition is invariant under rotations and translations of the tetrahedron and scaling amounts to changing $\ell$, which is invariant.
\end{proof}

\begin{remark}
	Note that we have only used Proposition \ref{prop:energylocallower} and the lower bound from Lemma \ref{lem:lowerrough}. This means that we could define the function $\elda$ for a larger class of interaction potentials. However, it is not clear if this function is then the same if one does not take tetrahedras as the underlying set. We will see how to prove this in our case in the rest of this chapter.
\end{remark}

Now that we have defined $\elda$, we can derive some global bounds on it, where we use the known estimates from section \ref{chapter:known_bds}.

\begin{lemma}[Rough bounds on $\elda$]
\label{lem:euegrough}
	There exists a constant $C > 0$ such that for any $\rho_0 > 0$, we have
	\begin{align}
		  i) & \quad \elda(\rho_0) \leq C \rho_0^{5/3}, \\
		 ii) & \quad \elda(\rho_0) \geq - C \rho_0, \\
		iii) & \quad \elda(\rho_0) \geq - C \rho_0^2. 
	\end{align}
\end{lemma}

\begin{proof}
	Parts $ii)$ and $iii)$ are a straight-forward application of Lemma \ref{lem:lowerrough}. For part $i)$, we notice that for any fixed $\eta$ in the definition of $\elda(\rho_0)$
	\begin{equation*}
		\lim_{\substack{\ell \to \infty}} \frac{1}{|\ell\Delta|} \intR \left| \nabla \sqrt{\rho_0 (\indvar{\ell\Delta}{})} \right|^2 = 0.
	\end{equation*}
	This is easy to see, as the integral is actually constant as long as $\ell$ is large enough. We then conclude by Lemma \ref{lem:upperrough}.
\end{proof}

The previous Lemma gives a global characterization of $\elda$. We also want a local characterization, i.e. a notion of continuity. It will turn out that under our assumptions on $\omega$, $\elda$ is locally Lipschitz. For this, we are interested in the behaviour of $E(\rho)$ under the unitary transformation $\rho \mapsto \alpha\rho(\alpha^{1/3}\cdot)$. This is done in the following two Lemmata.

\begin{lemma}[Energy scaling I]
\label{lem:scalingupper}
	For all densities $\rho$ and $\alpha \in (0,1]$, we have
	\begin{equation}
		E \left( \alpha \rho(\alpha^{1/3} \cdot) \right) \leq E(\rho) +  \frac{(1 - \alpha)\intR \omega}{2} \intR \rho^2
	\end{equation}
\end{lemma}

\begin{proof}
	By our assumption on $\omega$ for all $x \in \Rbb^3$ we have $\omega(\alpha^{-1/3}x) \leq \omega(x)$. This means that the interaction potential $\tilde{\omega} = \omega - \omega(\alpha^{-1/3} \cdot)$ is positive and by Lemma \ref{lem:intomega} $ii)$
	\begin{equation*}
		\Ccal_{\omega(\alpha^{-1/3} \cdot)}(\Gamma) - D_{\omega(\alpha^{-1/3} \cdot)}(\rho_\Gamma) \leq \Ccal_\omega(\Gamma) - D_\omega(\rho_\Gamma) + \frac{(1 - \alpha)\intR \omega}{2} \intR \rho_\Gamma^2
	\end{equation*}
	for all admissible $\Gamma$. Now by scaling
	\begin{equation*}
	\begin{split}
		E \left( \alpha \rho(\alpha^{1/3} \cdot) \right) & = \inf_\Gamma \left\{ \alpha^{2/3} \Tcal(\Gamma) + \Ccal_{\omega(\alpha^{-1/3} \cdot)}(\Gamma) - D_{\omega(\alpha^{-1/3} \cdot)}(\rho) \right\}\\
		& \leq \inf_\Gamma \left\{ \Tcal(\Gamma) + \Ccal_\omega(\Gamma) - D_\omega(\rho) \right\} +\frac{(1 - \alpha)\intR \omega}{2} \intR \rho^2 \\
		& = E(\rho) +  \frac{(1 - \alpha)\intR \omega}{2} \intR \rho^2.
	\end{split}
	\end{equation*}
\end{proof}

For the lower bound, we use a similar argument.

\begin{lemma}[Energy scaling II]
\label{lem:scalinglower}
	For all densities $\rho$ and $\alpha \in (0,1]$, we have
	\begin{equation}
		\frac{1}{\alpha} E \left( \alpha \rho(\alpha^{1/3} \cdot) \right) \geq E(\rho) -  \frac{\omega(0)}{2} \left(\frac{1}{\alpha} - 1 \right) \intR \rho.
	\end{equation}
\end{lemma}

\begin{proof}
	This is similar to the proof to the upper bound. First, let $\alpha \in (0,1]$ and define $\tilde{\omega}(x) \coloneqq \frac{1}{\alpha} \omega(\alpha^{-1/3} x) - \omega(x)$.  Then, because $\widehat{\omega}$ is decreasing, the Fourier Transform of $\tilde{\omega}$ is positive. Consequently, by Lemma \ref{lem:intomega} $i)$
	\begin{equation*}
		\Ccal_{\frac{1}{\alpha}\omega(\alpha^{-1/3} \cdot)}(\Gamma) - D_{\frac{1}{\alpha}\omega(\alpha^{-1/3} \cdot)}(\rho_\Gamma) \geq \Ccal_\omega(\Gamma) - D_\omega(\rho_\Gamma) + \frac{\omega(0)}{2} \left(\frac{1}{\alpha} - 1 \right) \intR \rho_\Gamma
	\end{equation*}
	for all admissible $\Gamma$. Again, similarly to the upper bound
	\begin{equation*}
	\begin{split}
		\frac{1}{\alpha} E \left( \alpha \rho(\alpha^{1/3} \cdot) \right) & = \inf_\Gamma \left\{ \alpha^{-1/3} \Tcal(\Gamma) + \Ccal_{\frac{1}{\alpha}\omega(\alpha^{-1/3} \cdot)}(\Gamma) - D_{\frac{1}{\alpha}\omega(\alpha^{-1/3} \cdot)}(\rho) \right\}\\
		& \geq \inf_\Gamma \left\{ \Tcal(\Gamma) + \Ccal_\omega(\Gamma) - D_\omega(\rho) \right\} - \frac{\omega(0)}{2} \left(\frac{1}{\alpha} - 1 \right) \intR \rho \\
		& = E(\rho) -  \frac{\omega(0)}{2} \left(\frac{1}{\alpha} - 1 \right) \intR \rho.
	\end{split}
	\end{equation*}
\end{proof}

While Lemma \ref{lem:euegex} gives us the existence of the function $\elda$ and a first characterization, we will now bootstrap this argument to first give a quantitative estimate on the convergence of the thermodynamic limit before later proving that $\elda$ actually doesn't need tetrahedra in its definition, one can take other (sufficiently regular) sets as well.

\begin{proposition}[Thermodynamic limit for tetrahedra]
\label{prop:ThLim}
	There exists a constant $C$ such that for all $0 < \rho_0$ and $\delta \leq \ell/C$ we have the upper bound
	\begin{equation}
	\label{eq:ThLim0}
		e_\Delta(\rho_0,\ell,\delta) \leq \elda(\rho_0) + C\frac{\rho_0}{\ell \delta} + C\frac{\rho_0^2(1+\delta)}{\ell}
	\end{equation}
	and the averaged lower bound for $\alpha \in (0,1/2)$
	\begin{equation}
	\label{eq:ThLim1}
		\left(\int_{1-\alpha}^{1+\alpha} \frac{ds}{s^4} \right)^{-1} \int_{1-\alpha}^{1+\alpha} e_\Delta(\rho_0,t\ell,t\delta) \frac{dt}{t^4} \geq \elda(\rho_0) - C\rho_0^2\delta^2 - C\frac{\delta}{\ell}\left( \rho_0+\rho_0^2 \right) - \frac{C}{\ell^2}\rho_0.
	\end{equation}
	Additionally
	\begin{equation}
		\lim_{\substack{\delta/\ell \to 0 \\ \ell\delta \to \infty}} e_\Delta(\rho_0,\ell,\delta) = \elda(\rho_0).
	\end{equation}
\end{proposition}

\begin{proof}
	Consider a large tetrahedron $\ell'\Delta$ smeared at a scale $\delta'$. We use Proposition \ref{prop:energylocallower} to relate its energy to the energy of smaller tetrahedrons of size $\ell \ll \ell'$, smeared at scale $\delta$:
	\begin{equation*}
	\begin{split}
		e_\Delta(\rho_0,\ell',\delta') = & \frac{E(\rho_0 \mathds{1}_{\ell'\Delta} \ast \eta_{\delta'})}{|\ell'\Delta|} \\
		\geq & \frac{1}{|\ell'\Delta||\ell\Delta|} \int_{SO(3)} \intR E \left( \rho_0 ((\indld)(R \cdot -z))(\mathds{1}_{\ell'\Delta} \ast \eta_{\delta'}) \right) dz \, dR \\
		& - \frac{C\rho_0}{\ell \delta} - \frac{C\rho_0^2(1+\delta)}{\ell}.
	\end{split}
	\end{equation*}
	Now the argument is similar to the one in Lemma \ref{lem:euegex}. Denote the set of $z$ such that $((\indld)(R \cdot -z))(\mathds{1}_{\ell'\Delta} \ast \eta_{\delta'}) = (\indld)(R \cdot -z))$ by $\Omega$. This set has measure which is smaller than $|\ell'\Delta|$ but it is always larger then a constant times the measure of the $(\ell + \delta + \delta')$-boundary of $\ell'\Delta$. For any $z \in \Omega$, we have
	\begin{equation*}
		E \left( \rho_0 ((\indld)(R \cdot -z))(\mathds{1}_{\ell'\Delta} \ast \eta_{\delta'}) \right) = E \left( \rho_0 ((\indld)(R \cdot -z)) \right) = |\ell\Delta|e_\Delta(\rho_0,\ell,\delta)
	\end{equation*}
	by the translation and rotation invariance of our energy. For $z$ such that $z \notin \Omega$ nor the energy is zero, we use
	\begin{equation*}
	\begin{split}
		E \left( \rho_0 ((\indld)(R \cdot -z))(\mathds{1}_{\ell'\Delta} \ast \eta_{\delta'}) \right) & \geq - \intR \rho_0 ((\indld)(R \cdot -z))(\mathds{1}_{\ell'\Delta} \ast \eta_{\delta'}) \\
		& \geq - \rho_0 \intR \indld \\
		& = - \rho_0 |\ell \Delta|
	\end{split}
	\end{equation*}
	where we used Lemma \ref{lem:lowerrough}. Again, the set where this happens has measure of order $(\ell + \delta + \delta')\ell'^2$. In total, we obtain
	\begin{equation}
	\label{eq:ThLim2}
		e_\Delta(\rho_0,\ell',\delta') \geq \frac{|\Omega|}{|\ell'\Delta|} e_\Delta(\rho_0,\ell,\delta) - C\rho_0 \frac{\ell+\delta+\delta'}{\ell'} -\frac{C\rho_0}{\ell \delta} - \frac{C\rho_0^2(1+\delta)}{\ell}.
	\end{equation}
	As in Lemma \ref{lem:euegex}, there are two cases: if $e_\Delta(\rho_0,\ell,\delta)$ is negative, we use $|\Omega| \leq |\ell'\Delta|$, if it is positive, we use $|\Omega| \geq |\ell'\Delta| - C(\ell + \delta + \delta')\ell'^2$. Now, as we have already shown that the limit exists, take $\ell'\to \infty$ at fixed $\ell,\delta,\delta'$ to get in both cases:
	\begin{equation*}
		\elda(\rho_0) \geq e_\Delta(\rho_0,\ell,\delta) -\frac{C\rho_0}{\ell \delta} - \frac{C\rho_0^2(1+\delta)}{\ell}
	\end{equation*}
	which is what we wanted to show.
	
	For the upper bound, again consider a tetrahedron $\ell'\Delta$ which is supposed to be large, $\ell$ small relative to $\ell'$ and $\delta \leq \ell/C$. Let $\alpha \in (0,1/2)$. Then, by  Proposition \ref{prop:energylocalupper}, with
	\begin{equation*}
		\chi_{\ell,\delta} \coloneqq \frac{1}{(1-\delta/\ell)^3}\mathds{1}_{\ell(1-\delta/\ell)\Delta} \ast \eta_\delta
	\end{equation*}
	we have
	\begin{equation*}
	\begin{split}
		e_\Delta(\rho_0,\ell',\delta') = & \frac{E(\rho_0 \mathds{1}_{\ell'\Delta} \ast \eta_{\delta'})}{|\ell'\Delta|} \\
		\leq & \left(\int_{1-\alpha}^{1+\alpha} \frac{ds}{s^4} \right)^{-1} \int_{1-\alpha}^{1+\alpha} \frac{dt}{t^4} \int_{SO(3)} dR \\
		& \frac{1}{|t\ell\Delta||\ell'\Delta|} \intR dz \; E(\rho_0 \chi_{t\ell,t\delta}(R\cdot -z) (\mathds{1}_{\ell'\Delta} \ast \eta_{\delta'})) + C \rho_0^2 \delta^2.
	\end{split}
	\end{equation*}
	Again denote by $\Omega$ the set of all $z$  where $\chi_{t\ell,t\delta}(R\cdot -z)(\mathds{1}_{\ell'\Delta} \ast \eta_\delta') = \chi_{t\ell,t\delta}(R\cdot -z) $. For $\Omega$, the same reasoning as for the lower bound applies: Its measure is always bounded  by $|\ell'\Delta|$ from above and by $|\ell'\Delta| - C(\ell + \delta + \delta')\ell'^2$ from below, where we already removed the $t-$dependence for simplicity. For $z \in \Omega$, by Lemma \ref{lem:scalinglower} it follows that
	\begin{equation*}
	\begin{split}
		 E(\rho_0 \chi_{t\ell,t\delta}&(R\cdot -z) (\mathds{1}_{\ell'\Delta} \ast \eta{_\delta'})) = E \left( \rho_0 \frac{1}{(1-\delta/\ell)^3} \mathds{1}_{t\ell(1-\delta/\ell)\Delta} \ast \eta_{t\delta} \right)\\
		 & \leq \frac{1}{(1-\delta/\ell)^3} E(\rho_0 \mathds{1}_{t\ell\Delta} \ast \eta_{t\delta/(1-\delta/\ell)}) + C\frac{\delta}{\ell} \rho_0 |t\ell\Delta| \\
		 & \leq E(\rho_0 \mathds{1}_{t\ell\Delta} \ast \eta_{t\delta/(1-\delta/\ell)}) + C\frac{\delta}{\ell} \left( \rho_0^{5/3}|t\ell\Delta| + \rho_0\frac{\ell^2}{\delta} \right) + C\frac{\delta}{\ell} \rho_0 |t\ell\Delta|,
	\end{split}
	\end{equation*}
	where the last inequality is obvious if the energy is negative. If it is positive, we use $(1-\delta/\ell)^{-3} \leq 1+C\delta/\ell$ and Lemma \ref{lem:upperrough}. Rewrite the error term as
	\begin{equation*}
		C\frac{\delta}{\ell} \left( \rho_0^{5/3}|t\ell\Delta| + \rho_0\frac{\ell^2}{\delta} \right) + C\frac{\delta}{\ell} \rho_0 |t\ell\Delta| \leq C\frac{\delta}{\ell} \left(\rho_0 + \rho_0^2 \right)|t\ell\Delta| + C\ell\rho_0
	\end{equation*}
	for a shorter notation.	For the set of $z$ such that $z \notin \Omega$ nor the energy is zero, we have a slightly more involved argument compared to the lower bound. Denote this set with a slight abuse of notation $\partial\Omega$. Now, for all $z \in \partial\Omega$, we have by Lemma \ref{lem:upperrough}
	\begin{multline*}
		E(\rho_0 \chi_{t\ell,t\delta}(R\cdot -z) (\mathds{1}_{\ell'\Delta} \ast \eta_{\delta'})) \leq C|t\ell\Delta|\rho_0^{5/3} + C\rho_0 \intR \left| \nabla \sqrt{\chi_{t\ell,t\delta}(Rx -z)} \right|^2 \; dx\\
		+ C \rho_0 \intR \chi_{t\ell,t\delta}(Rx -z) \left| \nabla \sqrt{\mathds{1}_{\ell'\Delta} \ast \eta_{\delta'}(x)} \right|^2 \; dx.
	\end{multline*}
	We bound the first integral by $\ell^2/\delta$ independently of $z$. We carry out the $dz$ integration on the last integral first, to get
	\begin{equation*}
	\begin{split}
		\frac{1}{|t\ell\Delta||\ell'\Delta|} \int_{\partial\Omega} &\intR \chi_{t\ell,t\delta}(Rx -z) \left| \nabla \sqrt{\mathds{1}_{\ell'\Delta} \ast \eta_{\delta'}(x)} \right|^2  \; dx \; dz\\
		&\leq \frac{1}{|t\ell\Delta||\ell'\Delta|} \intR \intR \chi_{t\ell,t\delta}(Rx -z) \left| \nabla \sqrt{\mathds{1}_{\ell'\Delta} \ast \eta_{\delta'}(x)} \right|^2 \; dz \; dx\\
		& = \frac{1}{|\ell'\Delta|} \intR \left|\nabla \sqrt{\mathds{1}_{\ell'\Delta} \ast \eta_{\delta'}(x)}\right|^2 \; dx\\
		& \leq \frac{C}{\ell'\delta'}.
	\end{split}
	\end{equation*}
	Now, using that $\partial\Omega$ is a set of order $(\ell+\delta+\delta')\ell'^2$, we get 
	\begin{multline*}
		e_\Delta(\rho_0,\ell',\delta') \\
		\leq \left(\int_{1-\alpha}^{1+\alpha} \frac{ds}{s^4} \right)^{-1} \int_{1-\alpha}^{1+\alpha} \frac{dt}{t^4} \frac{|\Omega|}{|\ell'\Delta|} \left( e_\Delta(\rho_0,t\ell,t\delta/(1-\delta/\ell)) + C\frac{\delta}{\ell}(\rho_0 + \rho_0^2) +C\frac{1}{\ell^2}\rho_0\right) \\
		+ C\frac{\ell+\delta + \delta'}{\ell'} \left( \rho_0^{5/3}+\rho_0/(\ell\delta) \right) + C \frac{\rho_0}{\ell'\delta'} + C\rho_0^2\delta^2.
	\end{multline*}
	As before, we take $\ell' \to \infty$ to obtain
	\begin{equation*}
		\left(\int_{1-\alpha}^{1+\alpha} \frac{ds}{s^4} \right)^{-1} \int_{1-\alpha}^{1+\alpha} e_\Delta(\rho_0,t\ell,t\delta/(1-\delta/\ell)) \frac{dt}{t^4} \geq \elda(\rho_0) - C\rho_0^2\delta^2 - C\frac{\delta}{\ell}(\rho_0+\rho_0^2) - C\frac{1}{\ell^2}\rho_0
	\end{equation*}
	and by taking $\delta = \tilde{\delta}/(1+\tilde{\delta}/\ell)$, we obtain (\ref{eq:ThLim1}).
	
	Now, for the convergence, equation (\ref{eq:ThLim0}) shows 
	\begin{equation*}
		\limsup_{\substack{\delta/\ell \to 0 \\ \ell\delta \to \infty}} e_\Delta(\rho_0,\ell,\delta)\leq \elda(\rho_0).
	\end{equation*}
	For the liminf, we take $\ell'$ and $\delta'$ with $\delta' / \ell' \to 0$ and $\ell' \delta' \to \infty$. We then use for $\ell \ll \ell'$  $e_\Delta(\rho_0,\ell,\delta) \leq C\rho_0^{5/3} + C\rho_0 /\delta\ell$ and (\ref{eq:ThLim2}) with the usual bounds on $\Omega$ to get
	\begin{equation*}
		e_\Delta(\rho_0,\ell',\delta') \geq e_\Delta(\rho_0,\ell,\delta) - C\frac{\ell+\delta+\delta'}{\ell'}\left(\rho_0^{5/3} + \rho_0/\delta\ell + \rho_0 \right) - C\frac{\rho_0}{\ell\delta} - C\frac{\rho_0^2(1+\delta)}{\ell}.
	\end{equation*}
	If we average this bound over $t \in (1-\alpha,1+\alpha)$ and use our lower bound (\ref{eq:ThLim1}), we have
	\begin{multline*}
		e_\Delta(\rho_0,\ell',\delta') \geq \elda(\rho_0) - C\frac{\ell+\delta+\delta'}{\ell'} \left(\rho_0^{5/3} + \rho_0/\delta\ell + \rho_0 \right) \\
		- C\frac{\rho_0}{\ell\delta} - C\frac{\rho_0^2(1+\delta)}{\ell} - C\rho_0^2\delta^2 - C\frac{\delta}{\ell}(\rho_0 + \rho_0^2) -C\frac{1}{\ell^2}\rho_0.
	\end{multline*}
	and we conclude by taking $\ell = (\ell')^{3/5}$ and $\delta = \ell^{-1/3}$.
\end{proof}

We are now ready to provide the

\begin{proof}[Proof of Theorem \ref{thm:uniformgas}]
	As in the previous proof, we use Proposition \ref{prop:energylocallower} to get
	\begin{multline*}
		\frac{E(\rho_0 \mathds{1}_{\Omega_N} \ast \eta_{\delta_N})}{|\Omega_N|} \geq \frac{1}{|\Omega_N||\ell\Delta|} \int_{SO(3)} \intR E(\rho_0 ((\indld)(R \cdot -z)) (\mathds{1}_{\Omega_N} \ast \eta_{\delta_N})) dz \, dR \\
		- C\frac{\rho_0}{\ell\delta} - C\frac{\rho_0^2(1+\delta)}{\ell}.
	\end{multline*}
	Again, denote the set where $\ell\Delta$ is truly inside $\Omega_N$ by $\Omega$. As before, its measure is less than $|\Omega_N|$ and more than $|\Omega_N|$ minus the measure of the $C(\ell+\delta+\delta_N)$-boundary of $\Omega_N$. Therefore by our assumption on the regularity of the boundary and the fact that $e_\Delta(\rho_0,\ell,\delta) \leq C(\rho_0^{5/3}+\rho_0/(\delta\ell))$, we have
	\begin{equation*}
		\frac{|\Omega|}{|\Omega_N|} e_\Delta(\rho_0,\ell,\delta) \geq e_\Delta(\rho_0,\ell,\delta) - C\frac{\ell+\delta+\delta_N}{|\Omega_N|^{1/3}}(\rho_0^{5/3}+\rho_0/(\delta\ell)).
	\end{equation*}
	With the same reasoning as before and exactly the same lower bound for tetrahedra close to the boundary it follows
	\begin{multline*}
		\frac{E(\rho_0 \mathds{1}_{\Omega_N} \ast \eta_{\delta_N})}{|\Omega_N|} \geq e_\Delta(\rho_0,\ell,\delta) - C\frac{\ell+\delta+\delta_N}{|\Omega_N|^{1/3}}(\rho_0^{5/3}+\rho_0/(\delta\ell)) - C\frac{\ell+\delta+\delta_N}{|\Omega_N|^{1/3}}\rho_0 \\
		- C\frac{\rho_0}{\ell\delta} - C\frac{\rho_0^2(1+\delta)}{\ell},
	\end{multline*}
	where we set for instance $\delta$ constant and $\ell = |\Omega_N|^{1/6}$ to get
	\begin{equation}
	\label{eq:uniformgas1}
		\liminf_{N \to \infty} \frac{E(\rho_0 \mathds{1}_{\Omega_N} \ast \eta_{\delta_N})}{|\Omega_N|} \geq \elda(\rho_0).
	\end{equation}
	
	For the lower bound, we use the same notation as in the previous Proposition and obtain:
	\begin{multline*}
		\frac{E(\rho_0 \mathds{1}_{\Omega_N} \ast \eta_{\delta_N})}{|\Omega_N|}	\leq \left(\int_{1-\alpha}^{1+\alpha} \frac{ds}{s^4} \right)^{-1} \int_{1-\alpha}^{1+\alpha} \frac{dt}{t^4} \int_{SO(3)} dR \\
		\frac{1}{|t\ell\Delta||\Omega_N|} \intR dz \; E(\rho_0 \chi_{t\ell,t\delta}(R\cdot -z) (\mathds{1}_{\Omega_N} \ast \eta_{\delta_N})) + C \rho_0^2 \delta^2.
	\end{multline*}
	Again, when the support of $\chi_{t\ell,t\delta}(R\cdot -z)$ is well inside the support of $\mathds{1}_{\Omega_N} \ast \eta_{\delta_N}$, we see as before
	\begin{equation*}
	\begin{split}
		E(\rho_0 \chi_{t\ell,t\delta}(R\cdot -z) (\mathds{1}_{\Omega_N} \ast \eta_{\delta_N})) & = E(\rho_0 \chi_{t\ell,t\delta}(R\cdot -z)) \\
		&\leq E(\rho_0 \mathds{1}_{t\ell\Delta} \ast \eta_{t\delta/(1-\delta/\ell)}) + C \frac{\delta}{\ell} \left(\rho_0 + \rho_0^2 \right) |t\ell\Delta| + C\ell\rho_0
	\end{split}
	\end{equation*}
	and with our usual case distinction and approximations of the measure of $\Omega$, we get
	\begin{multline*}
		\left(\int_{1-\alpha}^{1+\alpha} \frac{ds}{s^4} \right)^{-1} \int_{1-\alpha}^{1+\alpha} \frac{dt}{t^4} \frac{|\Omega|}{|\Omega_N|} \left\lbrace e_\Delta(\rho_0,t\ell,t\delta/(1-\delta/\ell)) + C\frac{\delta}{\ell} \left(\rho_0 + \rho_0^2 \right) + \frac{C}{\ell^2}\rho_0 \right\rbrace\\
		\leq \left(\int_{1-\alpha}^{1+\alpha} \frac{ds}{s^4} \right)^{-1} \int_{1-\alpha}^{1+\alpha} \frac{dt}{t^4} e_\Delta(\rho_0,t\ell,t\delta/(1-\delta/\ell)) + C\frac{\delta}{\ell}\left(\rho_0 + \rho_0^2 \right) + \frac{C}{\ell^2}\rho_0 \\
		C\frac{\ell+\delta+\delta_N}{|\Omega_N|^{1/3}}\rho_0.
	\end{multline*}
	For the tetrahedra at the boundary of $\Omega$, we get as in \cite{Lewin_2020} and before
	\begin{multline*}
		E(\rho_0 \chi_{t\ell,t\delta}(R\cdot -z) (\mathds{1}_{\Omega_N} \ast \eta_{\delta_N})) \\
		\leq C\ell^3 \left( \rho_0^{5/3}+\rho_0/(\ell\delta) \right) + C\rho_0 \intR \chi_{t\ell,t\delta}(R\cdot -z) \left|\nabla\sqrt{\mathds{1}_{\Omega_N} \ast \eta_{\delta_N}}\right|^2
	\end{multline*}
	by Lemma \ref{lem:upperrough}. For the left part of this upper bound, we use again that the measure of the boundary can be estimated by $C(\ell+\delta+\delta_N)|\Omega_N|^{2/3}$, for the right part, we first carry out the $dz$-integral over the whole $\Rbb^3$, using that the $\intR\chi = |t\ell\Delta|$. From this, it follows that
	\begin{equation*}
		\frac{1}{|\Omega_N|} \intR |\nabla \sqrt{\mathds{1}_{\Omega_N} \ast \eta_{\delta_N}} |^2 \leq \frac{C}{\delta_N |\Omega_N|^{1/3}}.
	\end{equation*}
	In total we find
	\begin{multline*}
		\frac{E(\rho_0 \mathds{1}_{\Omega_N} \ast \eta_{\delta_N})}{|\Omega_N|}	\leq \left(\int_{1-\alpha}^{1+\alpha} \frac{ds}{s^4} \right)^{-1} \int_{1-\alpha}^{1+\alpha} \frac{dt}{t^4} e_\Delta(\rho_0,t\ell,t\delta/(1-\delta/\ell)) + C\frac{\delta}{\ell}\left(\rho_0 + \rho_0^2 \right) \\
		+ \frac{C}{\ell^2}\rho_0 + C\frac{\ell+\delta+\delta_N}{|\Omega_N|^{1/3}}\rho_0 + C\frac{\ell+\delta+\delta_N}{|\Omega_N|^{1/3}} \left( \rho_0^{5/3}+\rho_0/(\ell\delta) \right) + C\frac{\rho_0}{\delta_N|\Omega_N|^{1/3}} + C\rho_0^2\delta^2.
	\end{multline*}
	Now, if we choose for instance $\ell = |\Omega_N|^{1/6}$ and $\delta = |\Omega_N|^{-1/12}$, we get
	\begin{equation}
	\label{eq:uniformgas2}
		\limsup_{N \to \infty} \frac{E(\rho_0 \mathds{1}_{\Omega_N} \ast \eta_{\delta_N})}{|\Omega_N|} \leq \elda(\rho_0).
	\end{equation}
	Equations (\ref{eq:uniformgas1}) and (\ref{eq:uniformgas2}) yield the desired result.
\end{proof}

\begin{lemma}[Upper local bound for $\elda$]
\label{lem:euegupper}
	There exist a constant $C$ such that for every $0 \leq \rho_1 < \rho_0$ we have
	\begin{equation}
		\elda(\rho_0 - \rho_1) \leq \elda(\rho_0) + C\rho_1 \rho_0.
	\end{equation}
\end{lemma}

\begin{proof}
	Let $\alpha \in (0,1]$. By Lemma \ref{lem:scalingupper}, we obtain
	\begin{equation*}
	\begin{split}
		\frac{e_\Delta(\alpha \rho_0, \ell/\alpha^{1/3}, \delta/\alpha^{1/3})}{\alpha} & = \frac{E(\alpha \rho_0 \mathds{1}_{\ell/\alpha^{1/3}\Delta} \ast \eta_{\delta/\alpha^{1/3}})}{\left| \ell\Delta \right|} \\
		& \leq e_\Delta(\rho_0, \ell, \delta) + C (1 - \alpha) \frac{\intR (\rho_0 \indld)^2}{\left| \ell\Delta \right|}.
	\end{split}
	\end{equation*}
	When passing to the limit using Proposition \ref{prop:ThLim} and taking $\alpha =(1 - \rho_1/\rho_0)$, we see that
	\begin{equation*}
		\elda(\rho_0 - \rho_1) \leq \elda(\rho_0) - \frac{\rho_1}{\rho_0} \elda(\rho_0) + C \rho_1 \rho_0.
	\end{equation*}
	We conclude by using Lemma \ref{lem:euegrough} $iii)$.
\end{proof}

Note that we could have used Lemma \ref{lem:euegrough} $ii)$ as well to get a slightly different error term. However, in our application it does not make much of difference, so we opted for this shorter error.

\begin{lemma}[Lower local bound for $\elda$]
\label{lem:eueglower}
	There exist a constant $C$ such that for every $0 \leq \rho_1 < \rho_0$ we have
	\begin{equation}
		\elda(\rho_0 - \rho_1) \geq \elda(\rho_0) - C\rho_1 \rho_0^{2/3} - C\rho_1.
	\end{equation}
\end{lemma}

\begin{proof}
	Again, let $\alpha \in (0,1]$. By Proposition \ref{prop:ThLim} and Lemma \ref{lem:scalinglower} we see that
	\begin{equation}
	\label{eq:eueglower1}
		\elda(\alpha \rho_0) \geq \alpha^2 \elda(\rho_0) - C (1 - \alpha) \rho_0.
	\end{equation}
	When taking $\alpha = (1-\eps)$ with $0 \leq \eps < 1$, this turns into
	\begin{equation*}
		\elda((1 - \eps) \rho_0) \geq (1 - \eps)^2 \elda(\rho_0) - C \eps \rho_0.
	\end{equation*}
	Now, when $\elda(\rho_0)$ is positive, we use $(1 - \eps)^2 \geq (1 - C\eps)$ and Lemma \ref{lem:euegrough} $i)$ to see $(1 - \eps)^2 \elda(\rho_0) \geq \elda(\rho_0) - C\eps \rho_0^{5/3}$.
	If $\elda(\rho_0)$ is negative, we can drop the pre-factor all-together. In both cases, we have
	\begin{equation*}
		\elda((1 - \eps) \rho_0) \geq \elda(\rho_0) - C\eps \rho_0^{5/3} - C \eps \rho_0.
	\end{equation*}
	and this time we conclude by taking $\eps = \rho_1/\rho_0$.
\end{proof}

We immediately obtain the following

\begin{corollary}[Lipschitz regularity of $\elda$]
\label{cor:eueglip}
	There exists a constant $C$ such that for all $0 < \rho_0$, $\rho_1$
	\begin{equation}
		\left| \elda(\rho_0) - \elda(\rho_1) \right| \leq C \left( \max(\rho_0, \rho_1) + 1 \right) \left| \rho_0 - \rho_1 \right|.
	\end{equation}
\end{corollary}
}

\section{Locally constant densities}\label{chapter:constden}
{
In this Section, we relate the energy of a localized density $E(\rho\indld)$ to the energy of the almost constant density $E(\rho_0\indld)$ where $\rho_0 = \rho(x)$ for some $x \in \supp \indld$. We start with a subadditivity estimate to relate the energy $E(\rho_1 + \rho_2)$ to $E(\rho_1)$ where $\rho_2$ is thought of as a small perturbation. Again, $\omega$ is a fixed short-range interaction.

\begin{lemma}[Rough subadditivity estimate]
\label{lem:subadd}
	Let $\rho_1$, $\rho_2 \in L^1(\mathbb{R}^3, \mathbb{R}_+)$ be two densities such that $\sqrt{\rho_1}$, $\sqrt{\rho_2} \in H^1(\mathbb{R}^3)$. Then there exits a constant $C > 0$ such that
	\begin{multline}
		E(\rho_1 + \rho_2) \leq E(\rho_1) + C\eps \intR \left( \rho_1^{5/3} + \rho_1 \right) + C\eps^{-2/3} \intR \rho_2^{5/3} \\
	 	+ C \intR {\left| \nabla \sqrt{\rho_2 + \eps\rho_1} \right|}^2 + \frac{1 - \eps}{\eps} D(\rho_2)
	\end{multline}
	for all $0 < \eps \leq 1$.
\end{lemma}

Note that we can estimate
\begin{equation*}
	\intR {\left| \nabla \sqrt{\rho_2 + \eps\rho_1} \right|}^2 \leq \intR {\left| \nabla \sqrt{\rho_2} \right|}^2 + \eps \intR {\left| \nabla \sqrt{\rho_1} \right|}^2
\end{equation*}
by the convexity of $\rho \mapsto {\left| \nabla \sqrt{\rho} \right|}^2$.

\begin{proof}[Proof of Lemma \ref{lem:subadd}]
	As in \cite{Lewin_2020}, we fix an $\eps \in (0,1]$ and consider two optimal states $\Gamma_1$ and $\Gamma_2$ for $\rho_1$ and $\rho_2/\eps + \rho_1$ respectively. Then
	\begin{equation*}
		\Gamma \coloneqq (1 - \eps)\Gamma_1 + \eps \Gamma_2
	\end{equation*}
	is a proper quantum state which has the density
	\begin{equation*}
		\rho_\Gamma = (1 - \eps)\rho_1 + \eps \left( \frac{\rho_2}{\eps} + \rho_1 \right) = \rho_1 + \rho_2.
	\end{equation*}
	By inserting this trial state and using the convexity of the grand-canonical Levy-Lieb functional and the upper bound from Lemma \ref{lem:upperrough} for $E(\rho_2/\eps + \rho_1)$, we see that
	\begin{equation*}
	\begin{split}
		E(\rho_1 + \rho_2) \leq &(1 - \eps)E(\rho_1) + C\eps \intR \left(\rho_1 + \frac{\rho_2}{\eps} \right)^{5/3} + C\eps \intR {\left| \nabla \sqrt{\rho_2/\eps + \rho_1} \right|}^2 \\
		& -D(\rho_1 + \rho_2) + (1 - \eps) D(\rho_1) + \eps D(\rho_1 + \rho_2/\eps).
	\end{split}
	\end{equation*}
	Now
	\begin{equation*}
	\begin{split}
		& -D(\rho_1 + \rho_2) + (1 - \eps) D(\rho_1) + \eps D(\rho_1 + \rho_2/\eps) \\
		= & -D(\rho_1) - 2D(\rho_1,\rho_2) - D(\rho_2) + (1 - \eps)D(\rho_1) + \eps D(\rho_1) + 2\eps D(\rho_1,\rho_2/\eps) + \frac{\eps}{\eps^2}D(\rho_2)\\
		= &\frac{1 - \eps}{\eps} D(\rho_2)
	\end{split}
	\end{equation*}
	and we conclude by using the lower bound from Lemma \ref{lem:lowerrough} which gives $-\eps E(\rho_1) \leq C \eps \intR \rho_1$.
\end{proof}

Before we can actually use our subadditivity estimate, we need a small technical lemma which can probably be found elsewhere, but we prove it here for completeness.

\begin{lemma}
\label{lem:sob}
	Let $u$ be continuous and $\nabla u \in L^p(U)$ for a bounded, open subset of $\Rbb^3$ with $C^1$ boundary and $p > 3$. Assume that $u$ vanishes at some point in $U$. Then there exists a constant which only depends on $U$ and $p$ such that
	\begin{equation}
		\lVert u \rVert_{L^\infty(\overline{U})} \leq C \lVert \nabla u \rVert_{L^p(U)}.
	\end{equation}
\end{lemma}

\begin{proof}
	Recall that by Morrey's inequality (see, for instance \cite[§5.6. Theorem 5]{Evans_2010}) we have
	\begin{equation}
	\label{eq:sob1}
		\lVert u \rVert_{C^{0,1-3/p}(\overline{U})} \leq C \lVert u \rVert_{W^{1,p}(U)},
	\end{equation}
	where $C$ is an universal constant which only depends on $U$ and $p$. The Hölder norm is defined by
	\begin{equation*}
		\lVert u \rVert_{C^{0,\gamma}(\overline{U})} \coloneqq \sup_{x \in U} |u(x)| + \sup_{\substack{x,y \in U \\ x \neq y}} \left\lbrace \frac{|u(x) - u(y)|}{|x-y|^\gamma} \right\rbrace.
	\end{equation*}
	This equation holds for all $u \in W^{1,p}(U)$, the Sobolev space of functions in $L^p$ for which all first-order (weak) partial derivatives are in $L^p$. We now claim that for functions which vanish at some point in $U$, we have the Poincaré-type inequality
	\begin{equation*}
		\lVert u \rVert_{L^p(U)} \leq C \lVert \nabla u \rVert_{L^p(U)}
	\end{equation*}
	for which the constant only depends on $U$ and $p$. Indeed, assume that there is no such constant, i.e. that for all $n$, there exist functions $u_n$ such that
	\begin{equation*}
		\lVert u_n \rVert_{L^p(U)} \geq n \lVert \nabla u_n \rVert_{L^p(U)}.
	\end{equation*}
	By normalizing, we can achieve $\lVert u_n \rVert_{L^p(U)} = 1$ and $\lVert \nabla u_n \rVert_{L^p(U)} \leq 1/n$. In particular, the $u_n$ form a bounded sequence in $W^{1,p}(U)$ and by (\ref{eq:sob1}), this sequence is bounded in Hölder norm and consequently has a convergent subsequence by Arzela-Ascoli. We call the limit $u$ and see that (possibly after taking another subsequence) $\nabla u$ is the weak limit of $\nabla u_n$. This implies $\nabla u = 0$. Also, since the $L^\infty$ norm is stronger than any $L^p$ norm in bounded sets, $u$ is the strong limit of the $u_n$'s and therefore $\lVert u \rVert_{L^p(U)} = 1$. If we can find $x \in \overline{U}$ such that $u(x)=0$, we found our contradiction and therefore have proved the lemma.
	
	To see this, we pass to another subsequence such that the zeros of the $u_n$'s, which we call $x_n$, converge to some $x \in \overline{U}$. Now we write
	\begin{equation}
		|u(x)| = |u(x) - u_n(x_n)| \leq |u(x) - u_n(x)| + |u_n(x) - u_n(x_n)|
	\end{equation}
	where the first term can be made small by choosing $n$ large enough such that $\lVert u - u_n \rVert_{L^\infty(\overline{U})}$ is small and the second term is small by the uniform continuity of the $u_n$'s.
\end{proof}

From now on, denote by 
\begin{equation*}
		\rhob \coloneqq \min_{x \in \supp (\indld)} \rho(x) \quad \quad \textrm{and} \quad \quad \brho \coloneqq \max_{x \in \supp (\indld)} \rho(x)
\end{equation*}
	the minimal and maximal value of $\rho$ on the support of $\indld$. In our applications, it will always be clear which $\ell$ and $\delta$ are meant, so we do not denote it in our notation to keep it simple. Since we only assumed $\rho \in H^1(\Rbb^3)$ so far and this is not enough for continuity, we will need a stronger assumption on $\rho$. The next lemma will contain such an assumption and it will be an important tool in our proof of the main theorem.

\begin{lemma}
\label{lem:rhoalbet}
	Let $p > 3$, $0 < \theta$, $a \leq 1$, $1 \leq \alpha$ and $0 \leq \beta \leq 1$. Furthermore, assume $\alpha a \leq p$. Then there exists a constant $C > 0$ such that for $\delta \leq \ell/C$, we have
	\begin{equation}
	\label{eq:rhoalbet1}
		\intR \left(\rho - \rhob \right)^\alpha \rho^\beta (\indld) \leq C\eps^{\alpha-1} \left(\frac{\ell^p}{\eps^{\frac{p}{a}-1}} \int_{2\ell\Delta} \left| \nabla \rho^\theta \right|^p + \eps \intR \rho^{\frac{\alpha + \beta - \theta a \alpha}{1- \frac{a \alpha}{p}}} (\indld) \right)
	\end{equation}
	and
	\begin{equation}
	\label{eq:rhoalbet2}
		\intR \left(\brho - \rho \right)^\alpha \rho^\beta (\indld) \leq C\eps^{\alpha-1} \left(\frac{\ell^p}{\eps^{\frac{p}{a}-1}} \int_{2\ell\Delta} \left| \nabla \rho^\theta \right|^p + \eps \intR \brho^{\frac{\alpha + \beta - \theta a \alpha}{1- \frac{a \alpha}{p}}} (\indld) \right)
	\end{equation}
	for any $\eps > 0$.
\end{lemma}

\begin{proof}
	First consider (\ref{eq:rhoalbet1}). For all $0 < \theta \leq 1$, we have
	\begin{equation*}
		\rho - \rhob \leq C \left(\rho^\theta - \rhob^\theta \right)\rho^{1-\theta},
	\end{equation*}
	so therefore
	\begin{equation*}
		(\rho - \rhob)^\alpha \leq C \left( \rho^\theta - \rhob^\theta \right)^{\alpha a} \rho^{\alpha (1-\theta a)}.
	\end{equation*}
	Thus, we get
	\begin{equation}
	\label{eq:rhoalbet3}
	\begin{split}
		& \intR (\rho - \rhob)^\alpha \rho^\beta (\indld) \\
		\leq & C \intR \left(\rho^\theta -\rhob^\theta \right)^{\alpha a} \rho^{\alpha + \beta - \theta a \alpha} (\indld) \\
		= & C \lVert \rho^\theta - \rhob^\theta \rVert_{L^\infty(\ell \Delta + B_\delta)}^{p\frac{\alpha a}{p}} \intR \rho^{\alpha + \beta - \theta a \alpha} (\indld).
	\end{split}
	\end{equation}
	
	Now we use Lemma \ref{lem:sob} for the set with $C^1$ boundary $\Delta + B_{1/C}$ and a scaling argument to see for every continuous $u$ which vanishes at least in one point in $\ell \Delta + B_\delta$
	\begin{equation*}
		\lVert u \rVert^p_{L^\infty(\overline{\ell \Delta + B_\delta})} \leq \lVert u \rVert^p_{L^\infty(\overline{\ell (\Delta + B_{1/C})})} \leq C\ell^{p-3} \int_{\ell (\Delta + B_{1/C})} \left| \nabla u \right|^p \leq C\ell^{p-3} \int_{2\ell\Delta} \left| \nabla u \right|^p,
	\end{equation*}
	where we chose $C$ large enough such that $\Delta + B_{1/C} \subset 2\Delta$.
	
	When using Hölder on the last integral in equation (\ref{eq:rhoalbet3}) with $q = p/\alpha a$, which is allowed by our assumptions, we obtain
	\begin{equation}
	\begin{split}
		& \intR (\rho - \rhob)^\alpha \rho^\beta (\indld) \\
		\leq & C \left(\ell^{p-3} \int_{\ell \Delta + B_\delta} \left| \nabla \rho^\theta \right|^p \right)^\frac{\alpha a}{p} \left( \intR \rho^{\frac{\alpha + \beta - \theta a \alpha}{1 - \frac{\alpha a}{p}}} (\indld) \right)^{1 - \frac{\alpha a }{p}} \left(\intR \indld \right)^\frac{\alpha a}{p} \\
		\leq & C \left(\ell^{p} \int_{\ell \Delta + B_\delta} \left| \nabla \rho^\theta \right|^p \right)^\frac{\alpha a}{p} \left( \intR \rho^{\frac{\alpha + \beta - \theta a \alpha}{1 - \frac{\alpha a}{p}}} (\indld) \right)^{1 - \frac{\alpha a }{p}}.
	\end{split}
	\end{equation}
	We conclude by realizing that for $0 \leq \tilde{p} \leq 1$ and any $\gamma$
	\begin{equation}
		x^{\tilde{p}} y^{1-\tilde{p}} \leq C(\eps^{-\gamma/\tilde{p}} x + \eps^{\gamma/1-\tilde{p}} y).
	\end{equation}
	Here, we choose $\gamma = \alpha(1 - \alpha a/p)$ to get the desired result. For (\ref{eq:rhoalbet2}), we use $\rho \leq \brho$ and
	\begin{equation*}
		\brho - \rho \leq C \left(\brho^\theta - \rho^\theta \right)\brho^{1-\theta}
	\end{equation*}
	and then complete the proof using exactly the same arguments as for the first inequality.
\end{proof}

By using the previous two lemmata together, we can finally relate the energy of the cut-off density to a locally constant density in the following way:

\begin{proposition}
\label{prop:rholocalupper}
	Let $p > 3$ and $0 < \theta < 1$. Assume additionally
	\begin{equation*}
		1 \leq \frac{5p}{3p-5}(1 - \theta) \leq 2
	\end{equation*}
	and
	\begin{equation*}
		1 \leq \frac{2p}{p-2}(1 - \theta) \leq 2.
	\end{equation*}
	Then there exists a constant $C = C(p,\theta)$ such that, for $\delta \leq \ell/C$, we have
	\begin{equation}
	\begin{split}
		E\left(\rho (\indld) \right) \leq  & E \left( \rhob (\indld) \right) + C \eps \intR \left(\rho + \rho^2 \right) (\indld) \\
		& + C \intR \rho \left| \nabla \sqrt{\indld} \right|^2 + \frac{C}{\eps} \intR \left| \nabla \sqrt{\rho} \right|^2 (\indld) \\
		& + C \frac{\ell^p}{\eps^{p-1}} \int_{2\ell \Delta} \left| \nabla \rho^\theta \right|^p
	\end{split}
	\end{equation}
	and
	\begin{equation}
	\begin{split}
		E\left(\rho (\indld) \right) \geq  & E \left( \brho (\indld) \right) - C \eps \ell^3(\brho + \brho^2) - \frac{C\ell^2}{\delta}\brho\\
		& - \frac{C}{\eps} \intR \left| \nabla \sqrt{\rho} \right|^2 (\indld) \\
		& - C \frac{\ell^p}{\eps^{p-1}} \int_{2\ell \Delta} \left| \nabla \rho^\theta \right|^p
	\end{split}
	\end{equation}
	for all $0 < \eps < 1$.
\end{proposition}

\begin{proof}
	We write $\rho = \rhob + (\rho - \rhob)$ and apply Lemma \ref{lem:subadd}. This gives
	\begin{equation}
	\begin{split}
		E\left(\rho (\indld) \right) \leq & E \left( \rhob (\indld) \right) +  C\eps \intR \left( \rho^{5/3} + \rho \right) (\indld) \\
		& + \frac{C}{\eps^{2/3}} \intR (\rho - \rhob)^{5/3} (\indld) + \frac{1}{\eps} D \left( (\rho - \rhob) (\indld) \right) \\
		& + C \intR \left| \nabla \sqrt{(\indld)(\rho - (1 - \eps)\rhob)} \right|^2.
	\end{split}
	\end{equation}
	In the first line, we have used that $\rhob \leq \rho$ on the support of $\indld$ and $\indld \leq 1$. First, we bound $\rho^{5/3}$ by $C \left(\rho + \rho^2 \right)$. Then, we follow \cite{Lewin_2020} by using
	\begin{equation*}
		\left| \nabla \sqrt{fg} \right|^2 = \frac{\left| \nabla (fg) \right|^2}{4fg} \leq \frac{f\left| \nabla g \right|^2}{2g} + \frac{g\left| \nabla f \right|^2}{2f}
	\end{equation*}
	and seeing that $\nabla (\rho - (1-\eps)\rhob) = \nabla \rho = 2\sqrt{\rho}\nabla \sqrt{\rho}$. Therefore, we have pointwise
	\begin{equation*}
		\left| \sqrt{(\indld)(\rho - (1-\eps)\rhob)} \right|^2 \leq (\indld)\frac{2\rho\left|\nabla \sqrt{\rho} \right|^2}{\rho - (1-\eps)\rhob} + 2\rho \left| \nabla \sqrt{\indld} \right|^2.
	\end{equation*}
	Then, from $\rho \geq \rhob$ it follows that $\eps \rho \leq \rho - (1-\eps)\rhob$, we have
	\begin{equation*}
		\frac{\rho}{\rho - (1-\eps)\rhob} \leq \frac{1}{\eps}
	\end{equation*}
	and thus
	\begin{equation*}
		\intR \left| \nabla \sqrt{(\indld)(\rho - (1 - \eps)\rhob)} \right|^2 \leq \frac{C}{\eps} \intR \left| \nabla \sqrt{\rho} \right|^2 (\indld) + C \intR \rho \left| \nabla \sqrt{\indld} \right|^2.
	\end{equation*}
	Now, by using Lemma \ref{lem:rhoalbet} with $a = 1$, $\alpha = 5/3$ and $\beta = 0$, we get
	\begin{equation*}
	\frac{C}{\eps^{2/3}} \intR (\rho - \rhob)^{5/3} (\indld) \leq C \frac{\ell^p}{\eps^{p-1}} \int_{2\ell \Delta} \left| \nabla \rho^\theta \right|^p + C \eps \intR \rho^{\frac{5p}{3p-5}(1 - \theta)} (\indld)
	\end{equation*}
	where we can bound the last term by an integral of the form $\intR \left(\rho + \rho^2 \right) (\indld)$ by our assumptions on $p$ and $\theta$. Similarly, we can bound our interaction term
	\begin{equation*}
	\begin{split}
		D \left( (\rho - \rhob) (\indld) \right) & \leq C \left\lVert \left(\rho - \rhob \right)\left(\indld \right) \right\rVert^2_{L^2} \\
		& \leq C \intR (\rho - \rhob)^2(\indld) \\
		& \leq C \eps \left( \frac{\ell^p}{\eps^{p-1}} \int_{2\ell \Delta} \left| \nabla \rho^\theta \right|^p + \eps \intR \rho^{\frac{2p}{p-2}(1 - \theta)} (\indld) \right).
	\end{split}
	\end{equation*}
	Here, we used that $\omega$ is short-range in the first line, $\indld \leq 1$ in the second line and Lemma \ref{lem:rhoalbet} with $a = 1$, $\alpha = 2$ and $\beta = 0$ in the last inequality. We conclude by bounding our last term as before.
	For our lower bound, writing $\brho = \rho + (\brho -\rho)$ in Lemma \ref{lem:subadd} gives
	\begin{equation}
	\begin{split}
		E\left(\brho (\indld) \right) \leq & E \left( \rho (\indld) \right) +  C\eps \intR \left( \rho^{5/3} + \rho \right) (\indld) \\
		& + \frac{C}{\eps^{2/3}} \intR (\brho - \rho)^{5/3} (\indld) + \frac{1}{\eps} D \left( (\brho - \rho) (\indld) \right) \\
		& + C \intR \left| \nabla \sqrt{(\indld)(\brho - (1 - \eps)\rho)} \right|^2.
	\end{split}
	\end{equation}
	We now use $\rho \leq \brho$, inequality (\ref{eq:rhoalbet2}) and estimate the direct interaction term as before. Also, with the same arguments as before and $\eps\rho \leq \brho - (1-\eps)\rho$ we estimate
	\begin{equation*}
		\intR \left| \nabla \sqrt{(\indld)(\brho - (1 - \eps)\rho)} \right|^2 \leq \frac{2}{\eps} \intR \left| \nabla \sqrt{\rho} \right|^2 (\indld) + \frac{C\ell^2}{\delta} \brho.
	\end{equation*}
	We conclude using our assumptions on the parameter $p$ and $\theta$.
\end{proof}
}

\section{Proof of the main Theorem}\label{chapter:mainthm}
{
We are now able to prove our main Theorem.

\begin{proof}[Proof of Theorem \ref{thm:lda}]
	In Section \ref{chapter:known_bds} we saw that
	\begin{equation*}
		|E(\rho)| \leq \clt(1+\eps) \intR \rho^{5/3} + \frac{\omega(0)}{2} \intR \rho + \frac{C(1+\eps)}{\eps} \intR |\nabla \sqrt{\rho} |^2.
	\end{equation*}
	Consequently, we have
	\begin{equation*}
		|\elda(\rho_0)| \leq \clt\rho_0^{5/3} + \frac{\omega(0)}{2}\rho_0.
	\end{equation*}
	Note that we could have written the error term with $\intR \omega$ as well, since we assumed $\omega$ and $\hat{\omega}$ to be positive.
	In both cases, we see that inequality (\ref{eq:lda}) is obvious for large $\eps$, so we only have to consider $\eps$ small. This is done in the following two Propositions.
\end{proof}

\begin{proposition}[Upper bound for small $\eps$]
\label{prop:energyupper}
	Let $\omega$ be a short-range interaction and $\eps > 0$ small. Assume additionally
	\begin{equation*}
		2 \leq p\theta \leq \frac{2}{5}p + 1.
	\end{equation*}
	Then there exists a constant $C$ such that
	\begin{equation}
		E(\rho) \leq \intR \elda(\rho) + C \eps \intR (\rho + \rho^2) + \frac{C}{\eps} \intR |\nabla \sqrt{\rho}|^2 + \frac{C}{\eps^{5/2p-1}} \intR |\nabla \rho^\theta |^p.
	\end{equation}
	for any $\rho \geq 0$.
\end{proposition}

\begin{proof}
	Take $\ell = \eps^{-3/2}$ and $\delta = \sqrt{\eps}$. Then we have for $0 < \beta < 1/2$, by Proposition \ref{prop:energylocalupper}
	\begin{multline}
	\label{eq:energylocalupper1}
		E(\rho) \leq \left(\int^{1+\beta}_{1-\beta} \frac{ds}{s^4} \right)^{-1} \int_{1/2+\beta}^{3/2-\beta} \frac{dt}{t^4} \int_{SO(3)} dR \int_{C_{t\ell}} \frac{d\tau}{(t\ell)^3} \times \\
		\times \sum_{z \in \Zbb^3} \sum_{j = 1}^{24} E \left( \chi_{t\ell,t\delta,j}(R \cdot - t\ell z - \tau)\rho \right) +  C \eps \intR \rho^2.
	\end{multline}
	where $\chi_{\ell, \delta, j} \coloneqq (1-\eps^2)^{-3} \mathds{1}_{\ell\mu_j (1-\eps^2)\Delta} \ast \eta_\delta$. Note that $(t\delta)/(t\ell) = \delta/\ell = \eps^2$.
	For simplicity of notation, we will derive a bound for $E(\tilde{\rho} (\mathds{1}_{\tilde{\ell} \Delta} \ast \eta_{\tilde{\delta}})$ where $\Delta$ is a tetrahedron of volume 1/24, but it can be anywhere in $\Rbb^3$, with any rotation. Then we can use this bound in our equation (\ref{eq:energylocalupper1}) with $\tilde{\rho} = (1 - \eps^2)^{-3}\rho$, $\tilde{\ell} = t\ell(1-\eps^2)$ and $\tilde{\delta} = t\delta$. Now, by Proposition \ref{prop:rholocalupper}, we have by our assumptions on $p$ and $\theta$
	\begin{equation*}
	\begin{split}
		E(\tilde{\rho} (\mathds{1}_{\tilde{\ell} \Delta} \ast \eta_{\tilde{\delta}})) \leq & E(\underline{\tilde{\rho}} (\mathds{1}_{\tilde{\ell} \Delta} \ast \eta_{\tilde{\delta}})) + C \eps \intR \left(\tilde{\rho} + \tilde{\rho}^2 \right) (\mathds{1}_{\tilde{\ell} \Delta} \ast \eta_{\tilde{\delta}}) \\
		& + C \intR \tilde{\rho} \left| \nabla \sqrt{\mathds{1}_{\tilde{\ell} \Delta} \ast \eta_{\tilde{\delta}}} \right|^2 + \frac{C}{\eps} \intR \left| \nabla \sqrt{\tilde{\rho}} \right|^2 (\mathds{1}_{\tilde{\ell} \Delta} \ast \eta_{\tilde{\delta}}) \\
		& + C \frac{\tilde{\ell}^p}{\eps^{p-1}} \int_{2\tilde{\ell} \Delta} \left| \nabla \tilde{\rho}^\theta \right|^p.
	\end{split}
	\end{equation*}
	Recall that for any $\rho$, $\ell$, $\delta$, we defined $\rhob \coloneqq \min_{x \in \supp (\indld)} \rho(x)$. Consider the first part of the last inequality. With the inequality from Proposition \ref{prop:ThLim}, we have
	\begin{equation*}
	\begin{split}
		E(\underline{\tilde{\rho}} (\mathds{1}_{\tilde{\ell} \Delta} \ast \eta_{\tilde{\delta}})) & = e_\Delta(\underline{\tilde{\rho}}, \tilde{\ell}. \tilde{\delta}) |\tilde{\ell}\Delta| \\
		& = e_\Delta(\underline{\tilde{\rho}}, \tilde{\ell}, \tilde{\delta}) \intR \mathds{1}_{\tilde{\ell} \Delta} \ast \eta_{\tilde{\delta}} \\
		& \leq \elda(\underline{\tilde{\rho}}) \intR \mathds{1}_{\tilde{\ell} \Delta} \ast \eta_{\tilde{\delta}} + C\eps(\underline{\tilde{\rho}} + \underline{\tilde{\rho}}^2)  \intR \mathds{1}_{\tilde{\ell} \Delta} \ast \eta_{\tilde{\delta}}
	\end{split}
	\end{equation*}
	In the support of $\mathds{1}_{\tilde{\ell} \Delta} \ast \eta_{\tilde{\delta}}$, it follows by Lemma \ref{lem:euegupper}
	\begin{equation*}
		\elda(\underline{\tilde{\rho}}) \leq \elda(\underline{\tilde{\rho}}(x)) + C(\tilde{\rho}(x) - \underline{\tilde{\rho}})\tilde{\rho}(x)
	\end{equation*}
	hence
	\begin{equation*}
		\elda(\underline{\tilde{\rho}}) \intR \mathds{1}_{\tilde{\ell} \Delta} \ast \eta_{\tilde{\delta}} \leq \intR \elda(\tilde{\rho}) (\mathds{1}_{\tilde{\ell} \Delta} \ast \eta_{\tilde{\delta}}) + C \intR (\tilde{\rho} - \underline{\tilde{\rho}})\tilde{\rho} (\mathds{1}_{\tilde{\ell} \Delta} \ast \eta_{\tilde{\delta}}).
	\end{equation*}
	The last term can be bounded by Lemma \ref{lem:rhoalbet}, where we use that $1 \leq p(2-\theta)/(p-1) \leq 2$. In total, by putting the last inequalities together and using that $(1-\eps^2)^{-3}$ is close to 1 for the error terms, we get
	\begin{equation}
	\label{eq:energylocalupper2}
	\begin{split}
		E(\chi_{t\ell,t\delta,j} \rho) = & E(\tilde{\rho}(\mathds{1}_{\tilde{\ell} \Delta} \ast \eta_{\tilde{\delta}})) \\
		\leq & (1-\eps^2)^3 \intR \elda \left( (1-\eps^2)^{-3}\rho(x) \right) \chi_{t\ell,t\delta,j}(x) dx + C\eps \intR (\rho + \rho^2)\chi_{t\ell,t\delta,j} \\
		& + \frac{C}{\eps^{5/2p-1}} \int_{2t\ell(1-\eps^2)^{-3} \Delta} \left| \nabla \rho^\theta \right|^p + \frac{C}{\eps} \intR |\nabla \sqrt{\rho}|^2 \chi_{t\ell,t\delta,j} \\
		& + C \intR \rho |\nabla \sqrt{\chi_{t\ell,t\delta,j}} |^2.
	\end{split}
	\end{equation}
	Following \cite{Lewin_2020}, we insert (\ref{eq:energylocalupper2}) into (\ref{eq:energylocalupper1}) and sum over the tiling using equation (\ref{eq:incomp_unity}) to obtain
	\begin{multline}
		E(\rho) \leq (1-\eps^2)^3 \intR \elda \left( (1-\eps^2)^{-3}\rho(x) \right) dx + C\eps \intR (\rho + \rho^2) \\
		+ \frac{C}{\eps} \intR |\nabla \sqrt{\rho}|^2 + \frac{C}{\eps^{5/2p-1}} \intR \left| \nabla \rho^\theta \right|^p.
	\end{multline}
	When we sum over the sets $2t\ell(1-\eps^2)^{-3} \Delta$, we use that they only have finitely many intersections, resulting in a bigger constant in front of it. Furthermore, we have used our usual estimate
	\begin{multline*}
		\int_{C_{t\ell}} \frac{d\tau}{(t\ell)^3} \sum_{z \in \Zbb^3} \sum_{j=1}^{24} \intR \rho(x) \left| \nabla \sqrt{\chi_{t\ell,t\delta,j}(x-\ell z - \tau)} \right|^2 \; dx\\
		= \frac{24}{(t\ell)^3} \intR \intR \rho(x) \left| \nabla \sqrt{\chi_{t\ell,t\delta,j}(x-z)} \right|^2 \; dz \; dx \leq \frac{C}{\ell\delta} \intR \rho = C\eps \intR \rho.
	\end{multline*}
	By equation (\ref{eq:eueglower1}) in Lemma \ref{lem:eueglower} ,we get for $0 < \alpha = (1-\eps^2)^3 < 1$
	\begin{equation*}
		\alpha^2 \elda(\alpha^{-1}\rho(x)) \leq \elda(\rho(x)) + C(1-\alpha)\alpha^{-1} \rho(x).
	\end{equation*}
	Therefore,
	\begin{multline*}
		(1-\eps^2)^3 \intR \elda \left( (1-\eps^2)^{-3}\rho(x) \right) dx \leq \frac{1}{	(1-\eps^2)^3} \intR \elda(\rho(x))dx \\
		+ \frac{1-(1-\eps^2)^3}{(1-\eps^2)^6} \intR \rho(x)dx
	\end{multline*}
	where the pre-factor on the last term can be bounded by $C\eps$ and the first pre-factor can be bounded by $1+C\eps$. We conclude by using $C\eps \intR \elda(\rho) \leq C\eps \intR (\rho + \rho^2)$ which follows from Lemma \ref{lem:euegrough}.
\end{proof}

\begin{proposition}[Lower bound for small $\eps$]
\label{prop:energylower}
	Let $\omega$ be a short-range interaction and $0 < \eps$ small. Assume additionally
	\begin{equation*}
		2 \leq p\theta \leq \frac{2}{5}p + 1.
	\end{equation*}
	Then there exists a constant $C$ such that
	\begin{equation}
		E(\rho) \geq \intR \elda(\rho) - C \eps \intR (\rho + \rho^2) - \frac{C}{\eps} \intR |\nabla \sqrt{\rho}|^2 - \frac{C}{\eps^{5/2p-1}} \intR |\nabla \rho^\theta |^p.
	\end{equation}
	for $\rho \geq 0$.
\end{proposition}

\begin{proof}
	Again, we take $\delta = \sqrt{\eps}$ and $\ell = \eps^{-3/2}$. From Corollary \ref{cor:energylocallower}, we get for $0 < \beta < 1/2$
	\begin{multline}
	\label{eq:energylower1}
		E(\rho) \geq \left( \int_{1-\beta}^{1+\beta} \frac{ds}{s^4} \right)^{-1} \int_{1-\beta}^{1+\beta} \frac{dt}{t^4} \frac{1}{(t\ell)^3} \times \\
		\times \sum_{z \in \Zbb} \sum_{j = 1}^{24} \int_{SO(3)} \int_{C_{t\ell}} E \left( \xi_{t\ell,t\delta,j}(R \cdot - t\ell z -\tau)\rho \right) dR \; d\tau - C\eps \intR (\rho + \rho^2),
	\end{multline}
	where we already averaged over $t$. Recall that $\xi_{\ell,\delta,j} = \indvar{\ell\mu_j\Delta}{\delta}$. We want to proceed similarly to the upper bound, so we will prove the following estimate for any tetrahedra $\Delta$ of fixed size
	\begin{multline}
	\label{eq:energylower2}
		\left( \int_{1-\beta}^{1+\beta} \frac{ds}{s^4} \right)^{-1} \int_{1-\beta}^{1+\beta} \frac{dt}{t^4} \frac{1}{(t\ell)^3} E \left( \rho (\indldt) \right) \\
		\geq \left( \int_{1-\beta}^{1+\beta} \frac{ds}{s^4} \right)^{-1} \int_{1-\beta}^{1+\beta} \frac{dt}{t^4} \frac{1}{(t\ell)^3} \Bigg\{ \intR \elda(\rho(x)) (\indldt) dx 
		- C \intR \rho |\nabla \sqrt{\indldt} |^2 \\
		- C\eps \intR (\rho + \rho^2)(\indldt) -\frac{C}{\eps} \intR |\nabla \sqrt{\rho}|^2 (\indldt) -\frac{C}{\eps^{5/2p-1}} \int_{4\ell\Delta} | \nabla \rho^\theta |^p \Bigg\}.
	\end{multline}
	As before, the last integral over the larger set $4\ell\Delta$ is not a problem, it only affects the constant in front of this term when we sum over all tetrahedra. We do this after inserting (\ref{eq:energylower2}) in (\ref{eq:energylower1}). Then we conclude in the same way as for the upper bound.
	From now on, denote with an abuse of notation
	\begin{equation*}
		\rhob \coloneqq \min_{2\ell\Delta + B_{2\delta}} \rho, \quad \brho \coloneqq \max_{2\ell\Delta + B_{2\delta}} \rho
	\end{equation*}
	instead of taking the minimum and the maximum over the smaller set. For the remainder of the proof, we will need that $\brho \leq C \rho(x)$ because most of our bounds involve the maximum $\brho$. In the following, we will see why this is false only in certain cases which can be dealt with differently. First, recall that by Lemmata \ref{lem:lowerrough}, \ref{lem:upperrough} and \ref{lem:euegrough} we have
	\begin{multline*}
		\left| E(\rho(\indldt)) - \intR (\indldt)\elda(\rho(x))dx \right| \leq C \intR (\indldt)(\rho^2 + \rho^{5/3}) \\
		+ C \intR (\indldt)|\nabla \sqrt{\rho} |^2 + C \intR \rho |\nabla \sqrt{\indldt}|^2.
	\end{multline*}
	Hence, we have shown (\ref{eq:energylower2}) when
	\begin{equation*}
		\intR (\indldt)(\rho^2 + \rho^{5/3}) \leq C\eps \intR (\indldt)(\rho + \rho^2) +\frac{1}{\eps^{5/2p-1}} \int_{4\ell\Delta} | \nabla \rho^\theta |^p.
	\end{equation*}
	This is true if $\brho^{2/3} \leq C\eps$. Hence we may assume that $\brho \geq C\eps^{3/2}$ and
	\begin{equation*}
		\intR (\indldt)(\rho^2 + \rho^{5/3}) \geq \frac{1}{\eps^{5/2p-1}} \int_{4\ell\Delta} | \nabla \rho^\theta |^p.
	\end{equation*}
	because otherwise we would also be done. Now, again by Sobolev, we have
	\begin{equation*}
	\begin{split}
		(\brho^\theta - \rhob^\theta)^p & = \lVert \brho^\theta - \rho^\theta \rVert^p_{L^\infty(\overline{2\ell \Delta + B_{2\delta}})} \\
		& \leq C \ell^{p-3} \int_{4\ell\Delta} | \nabla \rho^\theta |^p \\
		& \leq C \ell^{p-3} \eps^{5/2p-1} \intR (\indldt)(\rho^2 + \rho^{5/3}) \\
		& \leq C \ell^p \eps^{5/2p-1} (\brho^2 + \brho^{5/3}).
	\end{split}
	\end{equation*}
	With $\brho \geq C\eps^{3/2}$, we are able to estimate
	\begin{equation*}
		\brho^{5/3} = \brho^{p\theta} \frac{1}{\brho^{p\theta - 5/3}} \leq C \brho^{p\theta} \frac{1}{\eps^{p\theta 3/2 - 5/2}}
	\end{equation*}
	and
	\begin{equation*}
		\brho^2 = \brho^{p\theta} \frac{1}{\brho^{p\theta - 2}} \leq C \brho^{p\theta} \frac{1}{\eps^{p\theta 3/2 - 3}} \leq C \brho^{p\theta} \frac{1}{\eps^{p\theta 3/2 - 5/2}}
	\end{equation*}
	where we used that $\eps$ is small and $2 \leq p \theta$. Now, if we take the $p$-th root and insert our choice $\ell = \eps^{-3/2}$, we obtain
	\begin{equation*}
		\brho^\theta - \rhob^\theta \leq C \brho^\theta \eps^{1 + \frac{3}{2p} + \frac{3\theta}{2}}.
	\end{equation*}
	By our main assumption $p\theta \leq \frac{2}{5}p + 1$, the exponent on $\eps$ is positive, so if we take $\eps$ small enough, we get
	\begin{equation*}
		\brho \leq C \rhob \leq C \rho(x)
	\end{equation*}
	on $2\ell\Delta + B_{2\delta}$. This is what we wanted and now we can prove (\ref{eq:energylower2}) by familiar arguments.
	Namely, by the same reasoning as in Proposition \ref{prop:rholocalupper}, but with $\brho$ being the maximum over a larger set, we get
	\begin{multline*}
		E(\rho(\indldt)) \geq E(\brho(\indldt)) - C \eps \ell^3(\brho + \brho^2) \\
		- \frac{C}{\eps} \intR |\nabla \sqrt{\rho} |^2 (\indldt) - \frac{C}{\eps^{5/2p-1}} \int_{4\ell\Delta} | \nabla \rho^\theta |^p.
	\end{multline*}
	Because of $\brho \leq C\rho$, we can estimate the second term by $C\eps \intR (\rho + \rho^2)(\indldt)$. Now we average our estimate over $t$ and get by Proposition \ref{prop:ThLim} for the first part
	\begin{equation*}
	\begin{split}
		\left( \int_{1-\beta}^{1+\beta} \frac{ds}{s^4} \right)^{-1} \int_{1-\beta}^{1+\beta} \frac{dt}{t^4} & \frac{1}{(t\ell)^3} E \left( \brho (\indldt) \right) = |\Delta| \left( \int_{1-\beta}^{1+\beta} \frac{ds}{s^4} \right)^{-1} \int_{1-\beta}^{1+\beta} \frac{dt}{t^4} e_\Delta(\brho,t\ell,t\delta) \\
		& \geq \left( \elda(\brho) - C\eps(\brho + \brho^2) \right) |\Delta|.
	\end{split}
	\end{equation*}
	For the last term, we use
	\begin{equation*}
		C\eps(\brho + \brho^2) |\Delta| \leq \left( \int_{1-\beta}^{1+\beta} \frac{ds}{s^4} \right)^{-1} \int_{1-\beta}^{1+\beta} \frac{dt}{t^4} \frac{1}{(t\ell)^3} C\eps \intR (\rho + \rho^2)(\indldt).
	\end{equation*}
	by $\brho \leq C\rho$. With the same reasoning
	\begin{equation*}
		\elda(\brho)|\Delta| = \left( \int_{1-\beta}^{1+\beta} \frac{ds}{s^4} \right)^{-1} \int_{1-\beta}^{1+\beta} \frac{dt}{t^4} \frac{1}{(t\ell)^3} \intR \elda(\brho)(\indldt).
	\end{equation*}	
	By this estimate and Lemma \ref{lem:euegupper}, we have
	\begin{equation*}
		\elda(\brho) \geq \elda(\rho(x)) - C(\brho - \rho(x))\rho(x)
	\end{equation*}
	on the support of $\indldt$. We estimate $\intR (\brho - \rho(x))\rho(x)(\indldt)$ using Lemma \ref{lem:rhoalbet}, where we again fulfil the assumptions on $p$ and $\theta$. This proves equation (\ref{eq:energylower2}) and thereby the Proposition.
\end{proof}
}


\bibliographystyle{acm}
\bibliography{./bib/biblio}

\end{document}